\documentclass[amsmath,amssymb,aps,prl,notitlepage,superscriptaddress,longbibliography,reprint]{revtex4-2}
\usepackage{graphicx}
\usepackage{dcolumn}
\usepackage{bm}
\usepackage{braket}
\usepackage{amsthm}
\usepackage{epstopdf}
\usepackage{amssymb}
\usepackage{amsmath}
\usepackage{bbold}
\usepackage{hyperref}
\usepackage{comment}
\usepackage{xcolor}
\usepackage[ruled,vlined]{algorithm2e}

\newcommand{\cor}[1]{\textcolor{red}{#1}}

\newcommand{\Haf}{\text{Haf}}
\newcommand{\Per}{\text{Per}}

\newcommand{\poly}{\text{poly}}

\theoremstyle{definition}

\newtheorem{theorem}{Theorem}

\newtheorem{lemma}{Lemma}

\begin{document}
\preprint{APS/123-QED}
\title{Classical simulation of boson sampling based on graph structure}
\author{Changhun Oh}%
\email{changhun@uchicago.edu}
\affiliation{Pritzker School of Molecular Engineering, University of Chicago, Chicago, Illinois 60637, USA}
\author{Youngrong Lim}%
\email{sshaep@kias.re.kr}
\affiliation{School of Computational Sciences, Korea Institute for Advanced Study, Seoul 02455, Korea}
\author{Bill Fefferman}
\email{wjf@uchicago.edu}
\affiliation{Department of Computer Science, University of Chicago, Chicago, Illinois 60637, USA}
\author{Liang Jiang}
\email{liang.jiang@uchicago.edu}
\affiliation{Pritzker School of Molecular Engineering, University of Chicago, Chicago, Illinois 60637, USA}
\date{\today}
\begin{abstract}
Boson sampling is a fundamentally and practically important task that can be used to demonstrate quantum supremacy using noisy intermediate-scale quantum devices.
In this work, we present classical sampling algorithms for single-photon and Gaussian input states that take advantage of a graph structure of a linear-optical circuit.
The algorithms' complexity grows as so-called treewidth, which is closely related to the connectivity of a given linear-optical circuit.
Using the algorithms, we study approximated simulations for local Haar-random linear-optical circuits. 
For equally spaced initial sources, we show that when the circuit depth is less than the quadratic in the lattice spacing, the efficient simulation is possible with an exponentially small error.
Notably, right after this depth, photons start to interfere each other and the algorithms' complexity becomes sub-exponential in the number of sources, implying that there is a sharp transition of its complexity.
Finally, when a circuit is sufficiently deep enough for photons to typically propagate to all modes, the complexity becomes exponential as generic sampling algorithms.
We numerically implement a likelihood test with a recent Gaussian boson sampling experiment and show that the treewidth-based algorithm with a limited treewidth renders a larger likelihood than the experimental data.
\end{abstract}

\maketitle
Sampling from the probability distributions of random quantum circuits is one of the problems to demonstrate quantum supremacy using noisy intermediate-scale quantum devices \cite{bremner2011classical, fefferman2015power, boixo2018characterizing, preskillnisq, arute2019quantum}.
Boson sampling (BS) is one such sampling problem using linear-optical devices believed to be hard to classically simulate under some plausible assumptions \cite{aaronson2011computational, hamilton2017gaussian}. 
While a scale of experimental BS grows rapidly due to its importance \cite{zhong2019experimental, zhong2020quantum, zhong2021phase}, classical simulation algorithms taking advantage of current BS experiments' limitations are still restricted.
Photon loss and distinguishability of photons are representative limitations, which have been extensively studied recently and shown to be detrimental to quantum advantages \cite{aaronson2016bosonsampling, oszmaniec2018classical, garcia2019simulating, renema2018classical, moylett2019classically, qi2020regimes, oh2021classical}.
Another limitation of current experiments is that the number of modes is not sufficiently large to reach a collision-free BS, which may also reduce the complexity of classical simulation \cite{clifford2020faster, bulmer2021boundary}.

In this Letter, we focus on limited connectivity of a linear-optical circuit.
In general, typical global Haar-random linear-optical circuits' input and output modes are fully connected, which makes it hard to classically simulate.
One possible implementation of global Haar-random circuits is to prepare local beam-splitter arrays \cite{emerson2005convergence}, which corresponds to the current BS experiments' setup.
However, a deviation from a global Haar-random unitary is apparent in the recent experiments \cite{zhong2019experimental, zhong2020quantum} because either the circuit depth is small or appropriate ensemble of beam splitters are not employed \cite{russell2017direct}.
Hence, there is a chance that the connectivity of the circuit is limited and that sampling from the underlying system may not be as difficult as from a global Haar-random circuit.



We propose classical algorithms using dynamical programming \cite{bodlaender2008combinatorial, cifuentes2016efficient} taking advantage of a given circuit's limited connectivity for single-photon BS (SPBS) and Gaussian BS (GBS) \cite{aaronson2011computational, hamilton2017gaussian}.
Particularly, our algorithms' complexity depends on connectivity of a relevant matrix's graph structure, characterized by the so-called treewidth \cite{cormen2009introduction}.
Since the algorithms' complexity grows as the treewidth instead of the system size, we may be able to sample from some linear-optical circuits of a limited treewidth faster than generic classical algorithms.
By applying our algorithm to local beam-splitter circuits, we analyze how the algorithms' complexity grows as a circuit depth and reveal a hierarchy of the complexity depending on the depth, namely, polynomial, sub-exponential, and exponential regimes. 

\emph{Boson sampling.---}\label{sec:sampling}
Consider an $M$-mode bosonic system consisting of beam-splitter arrays characterized by a unitary matrix $U$ with $N$ identical sources.
Specifically, the unitary matrix $U$ represents the transformation of mode operators $\{\hat{a}_j\}_{j=1}^M$ as
$\hat{a}_j^\dagger\to \hat{U}^\dagger\hat{a}_j^\dagger\hat{U}=\sum_{k=1}^M U_{jk} \hat{a}_k^\dagger$ for a given beam splitter circuit $\hat{U}$.
Let $\mathcal{S}\equiv \{s_i\}_{i=1}^N\subset [M]$ be the set of input modes for identical sources.
If we measure an output state $\hat{\rho}$ after beam splitters with the photon number basis $\hat{\bm m}=\otimes^M_{j=1}\ket{m_j}\bra{m_j}$, the probability of an outcome $\bm{m}=(m_1,...,m_M)$ is given by $P(\bm{m})=\text{Tr}(\hat{\rho}\hat{\bm m})$.
For simplicity, we define an equivalent description of the output as $\bm{r}=(r_1,\dots,r_N)$, where $r_i$'s represent modes that click.
For single-photon state input, the probability is written as~\cite{aaronson2011computational}
\begin{equation}\label{permanent}
    P(\bm{m})
    =\frac{|\Per(U^{\mathcal{S}}_{\bm{r}})|^2}{\bm{m}!}
    =\frac{1}{\bm{m}!}\left|\sum_\sigma \prod^N_{i=1} U_{r_{i},s_{\sigma(i)}}\right|^2,
\end{equation}
where the sum is over all permutations $\sigma$.
Here, $U^{\mathcal{S}}_{\bm{r}}$ is an $N\times N$ matrix obtained by choosing $\mathcal{S}$ columns and $\bm{r}$ rows, and $\Per(U)$ is the permanent of matrix $U$, which is related to counting bipartite perfect matchings in the corresponding graph~\cite{lovasz2009matching}.
Meanwhile, for a squeezed vacuum state input, the probability of an outcome $\bm{m}$ is given by~\cite{hamilton2017gaussian}
\begin{equation}\label{hafnian}
    P(\bm{m})=\frac{|\Haf(B_{\bm{m}})|^2}{\bm{m}!\sqrt{\det(V+\mathbb{1}/2)}},
\end{equation}
where $B_{\bm{m}}$ is a matrix obtained by repeating $i$th column and row of $B\equiv UDU^\text{T}$ for $m_i$ times, and $\Haf(B_{\bm{m}})$ is the hafnian of matrix $B_{\bm{m}}$, which is related to counting perfect matchings in the corresponding graph \cite{barvinok2016combinatorics}.
Here, $D\equiv \oplus_{j=1}^M\tanh{r_j}$, and $V$ is the output state's the covariance matrix.
Squeezing parameters are given by $r_j=r$ for $j\in \mathcal{S}$, and $r_j=0$ otherwise. 

Let us clarify the relation between graphs and BS (see Fig.~\ref{fig:graph} (a)).
To compute the (marginal) probability for an outcome, we consider all possible paths from input photons to the output configuration, which essentially corresponds to interference.
They can be described by all perfect matchings of a bipartite graph of $U_{\bm{r}}^{\mathcal{S}}$ with the input modes $\mathcal{S}$ and output modes $\bm{r}$ being bipartite vertex sets and the paths between them being edges for SPBS. 
For GBS, vertices of a symmetric graph of $B_{\bm{m}}$ consist of an output-photon configuration, and two vertices have an edge if the two photons can come from the same source.
To compute a probability in this case, we consider all possible perfect matchings of output photons, which corresponds to finding sources from which each pair of photons come.
From this observation, when a given unitary matrix's connectivity is limited, we can expect that the number of possible perfect matchings for each outcome is small so that the induced graphs' structure allows to reduce the complexity.

\begin{figure}[t]
\includegraphics[width=240px]{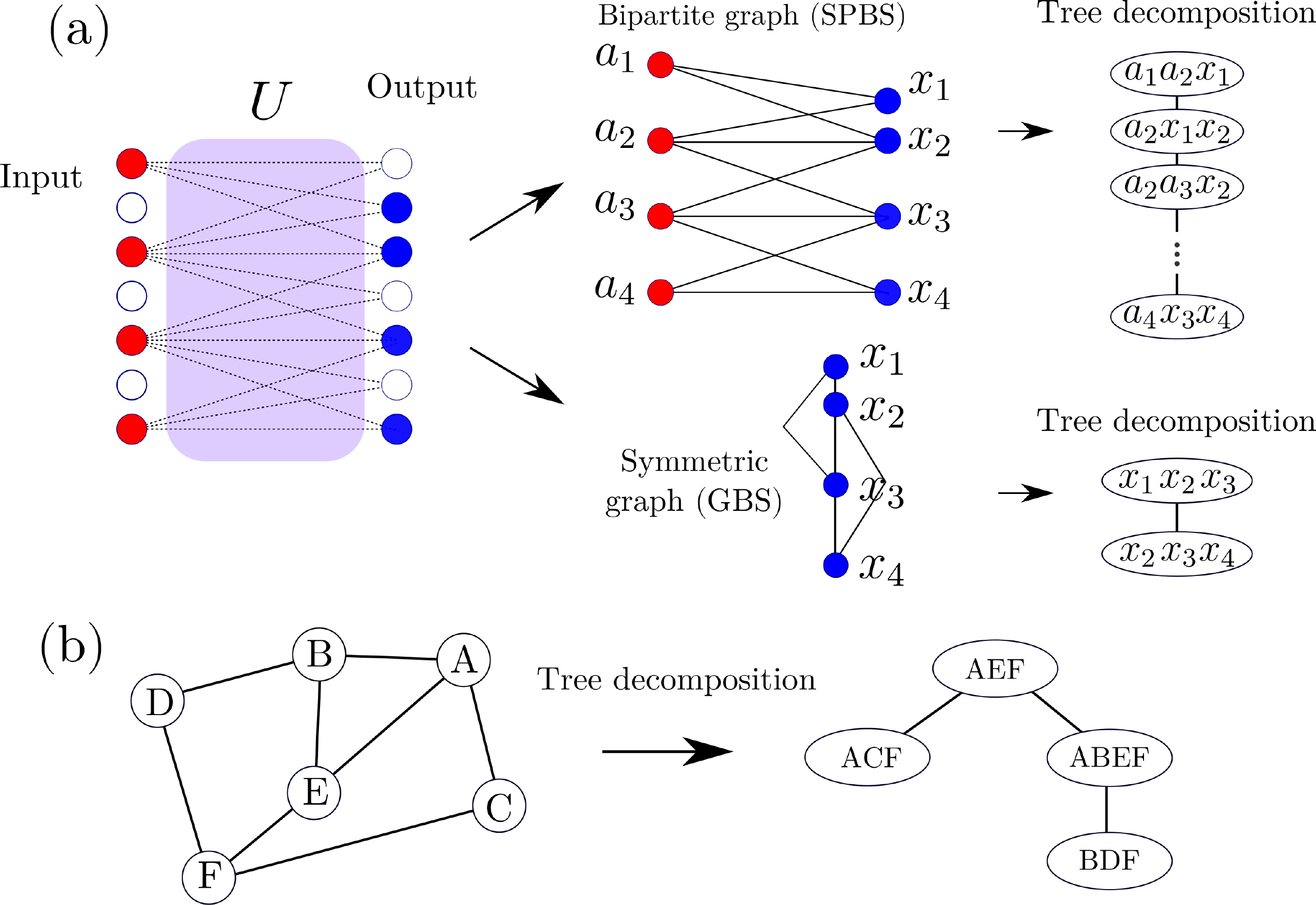}
\caption{(a) Input (red dots) and output (blue dots) photon configuration, corresponding bipartite and symmetric graphs and their tree decompositions of width $w=2$. (b) Graph and its possible tree decomposition of width $w=3$.}
\label{fig:graph}
\end{figure}

\emph{Computing permanent and loop hafnian using dynamical programming.---}\label{sec:treewidth}
Before presenting sampling algorithms, we first introduce classical algorithms computing the permanent and loop hafnian of a matrix.
Here, loop hafnian is generalized hafnian, related to counting perfect matchings including loops \cite{bjorklund2019faster, loophaf}, which is necessary for the sampling algorithm below.
The complexity of the best-known algorithms computing the permanent and loop hafnian of a general $N\times N$ matrix scales as $2^{N}$ and $2^{N/2}$, respectively ~\cite{ryser1963combinatorial,bjorklund2012counting}.
Meanwhile, there are also various algorithms exploiting a matrix's structures~\cite{barvinok1996two,bjorklund2019faster,schwartz2009efficiently}.
A particularly interesting algorithm is dynamical programming that computes permanent \cite{cifuentes2016efficient}.
A high-level idea of the algorithm is to construct tree decomposition of bipartite graph for a given matrix, which reveals the matrix's structure (see Fig.~\ref{fig:graph}).
The algorithm's complexity grows as so-called treewidth, which measures connectivity by exploiting the treelike structure of the graph~\cite{bodlaender2008combinatorial}.
We generalize the treewidth-based algorithm to loop hafnian by using tree decompositions for a given symmetric matrix and present the following lemma, including the result in Ref.~\cite{cifuentes2016efficient} as:
\begin{lemma}\label{lemma:treewidth}
    If the treewidth of a graph representation of an $N\times N$ matrix is $w$, then dynamical programming can compute its permanent and loop hafnian in $O(Nw^22^w)$.
\end{lemma}
We provide the proofs of Lemmas and Theorems in Ref.~\cite{supple}. 
Notably, Lemma~\ref{lemma:treewidth} shows that the complexity's exponent does not scale as the matrix size $N$ but the treewidth $w$. 
Therefore, for some structured matrices, the complexity of computing their permanent or loop hafnian can be highly reduced.
For example, a forest, i.e., disjoint union of trees, has treewidth 1~\cite{harary2018graph}, so the complexity does not grow exponentially as matrix size. 
On the other hand, a complete graph, whose vertices are all connected, has the treewidth $N-1$ ($N$ for bipartite complete graph)~\cite{harary2018graph}.
Note that we recover the same exponent of the algorithm for a general matrix, i.e., $2^N$, for permanent whereas it has a gap for loop hafnian ($2^{N/2}$ for general loop hafnian) \cite{bjorklund2012counting, bjorklund2019faster}.


\emph{Classical sampling algorithms based on treewidth.---}
We now introduce classical sampling algorithms of SPBS and GBS using limited connectivity. 
Although we have algorithms computing permanent or loop hafnian using a given graph's structure, how to use such algorithms for sampling is not clear.
Remarkably, we show that if we employ as a main routine chain rule of marginal probabilities, such as the Clifford-Clifford algorithm~\cite{clifford2018classical} for SPBS, and a recently proposed GBS algorithm \cite{quesada2022quadratic}, and use our dynamical programming to compute permanent or loop hafnian as a subroutine, we can fully utilize the graph structure of a given circuit including computing marginal probabilities~\cite{supple}.
For simplicity, we focus on collision-free events, i.e., $m_i=\{0,1\}$, while we provide algorithms for collisions in Ref.~\cite{supple}.
\begin{theorem} (Classical sampling algorithm) \label{theore:exact}
    If the treewidths of bipartite graphs of $U^{\mathcal{S}}_{\bm{r}}$ are at most $w$ for all possible outcome $\bm{r}$, we can classically simulate SPBS in $O(MN^2w^22^w)$.
    Similarly, if symmetric graphs of $B_{\bm{m}}$ have the treewidths at most $w$ for all outcomes $\bm{m}$, we can classically simulate GBS in $O(MNw^22^w)$.
\end{theorem}
Theorem~\ref{theore:exact} enables us to recover and generalize some previously known results.
One such example is efficient simulability of shallow 1D GBS, i.e., depth $D=O(\log M)$ by using limited bandwidth of the circuit's unitary matrix ~\cite{bandwidth, qi2020efficient}.
Since bandwidth is a special case of treewidth, we recover the result and also find that the result holds for 1D SPBS.
For 2D cases, however, even for a constant depth, we encounter with an output described by a graph including $\sqrt{N}\times \sqrt{N}$ grid, whose treewidth is $w=\sqrt{N}$ \cite{diestel2005graph, supple}.
This is consistent with the recent hardness result of high-dimensional GBS~\cite{deshpande2021quantum}.

\emph{Approximate sampling.---}
When an approximation of a given circuit has limited connectivity, we can expect that an approximate sampling is possible using this structure. 
However, it is not straightforward to apply the same method if we approximate the circuit matrix by a nonunitary matrix because the corresponding process or the output state may no longer be physical.
Also, the chain-rule-based algorithms implicitly assume unitarity of the process or a legitimate quantum state.
We present a method to overcome this by introducing additional virtual $M$ modes to make the process physical again and investigate its approximation error in Ref.~\cite{supple}:
\begin{theorem}(Approximate sampling)\label{theore:approx}
    If a circuit unitary matrix $U$ is approximated by $U-dU$, one can implement sampling with the same complexity up to constant as Theorem~\ref{theore:exact} with an error of $\poly(N,\|dU\|_F^{1/4})$.
\end{theorem}
We assess a simulation's error by total variation distance $\sum_{\bm m}|P(\bm{m})-P_a(\bm{m})|/2$ between an ideal probability distribution $P(\bm{m})$ and a classical algorithm's output probability distribution $P_a(\bm{m})$ and desire an approximation error to be $O(1/\poly(N))$.
In the following section, we study an experimentally relevant physical model, which is local Haar-random circuits.
Since a current GBS experiment does not employ a specialized ensemble to implement a global Haar-random circuit \cite{zhong2020quantum}, its setup can be considered as a typical instance of the model.
Also, it can be interpreted as an extreme case where beam splitters' reflectivities have a large uncertainty.
We emphasize that our approximation method in Theorem \ref{theore:approx} is straightforwardly applicable to similar dynamics (e.g. Ref. \cite{deshpande2018dynamical}).

\begin{figure}[t]
\includegraphics[width=220px]{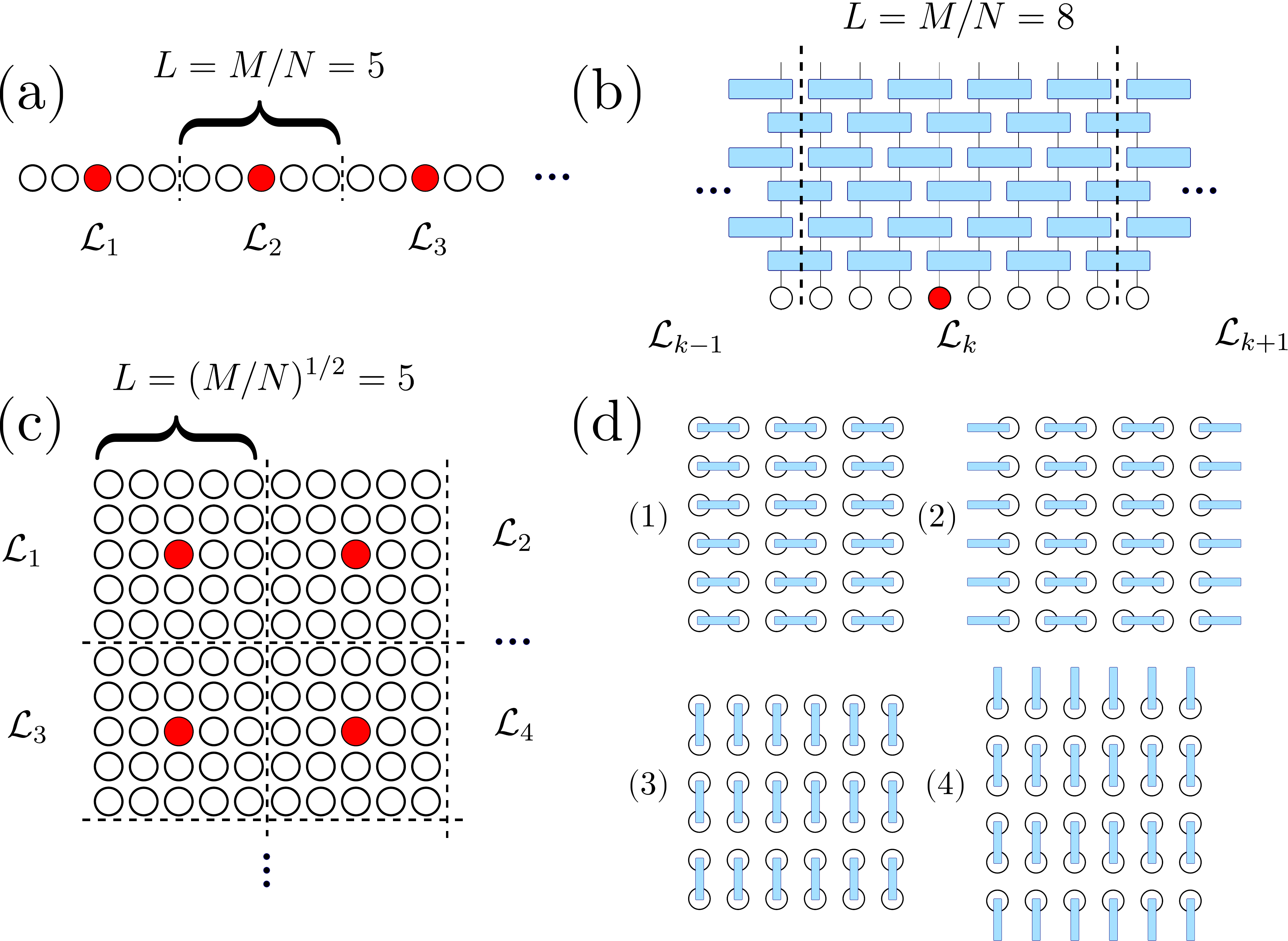}
\caption{Initial state in (a) 1D and (c) 2D architectures. Red dots represent sources.
$\mathcal{L}_\alpha$ represents a sublattice having a single source $s_\alpha$.
Beam-splitter arrays in (b) 1D and (d) 2D architecture. A single round consists of four steps (1)-(4). The structure can be generalized for $d$-dimensional architecture, where a single round consists of $2d$ steps.}
\label{fig:circuit}
\end{figure}

\emph{Approximate sampling for local Haar-random circuits.---}\label{sec:mapping}
Consider $N$ identical sources equally distributed in $M=kN^\gamma$ modes of a $d$-dimensional lattice \cite{equal} and local Haar-random beam-splitter arrays, as illustrated in Fig.~\ref{fig:circuit}.
The lattice consists of $d$-cube sublattices of edge length $L=(M/N)^{1/d}$, containing a single source.
For simplicity, let $L$ be a positive integer.



As recently studied, random beam-splitter arrays can be characterized by a classical random walk \cite{zhang2020entanglement}.
Therefore, photons propagate diffusively on average.
Using this property, we find an upper-bound on the leakage rate from a source at $s_\alpha$ up to $\kappa L$ denoted as $\eta_\alpha(\kappa)\equiv \sum_{j}|U_{j,s_\alpha}|^2$,
where $j$ is the sum over modes away from $\alpha$ more than $\kappa L$:
\begin{lemma}\label{lemma:leakage}
    For depth $D\leq dk^{2/d} \kappa^2N^{2(\gamma-1)/d-\epsilon}/2$ with $\epsilon>0$, the leakage rate $\eta_\alpha$ to distance $\kappa L$ is bounded from above as
    \begin{align}
        \eta_\alpha(\kappa)\leq \exp(-N^{\epsilon})
    \end{align}
    with a probability $1-\delta$ over Haar-random beam-splitter arrays, where $\delta$ is exponentially small in $N$.
\end{lemma}
For later usage for $d=1$, we note that the same inequality holds for $D\leq k^2\kappa^2N^{2(\gamma-1)-\epsilon}(\log N)^2/2$ for leakage rate to distance $\kappa L\log N$.
Motivated by Lemma~\ref{lemma:leakage}, our approximate sampling strategy is to discard the elements of a unitary matrix that are geometrically farther from sources than $\kappa L$, i.e., $U\to\tilde{U}\equiv U-dU$ and implement Theorem~\ref{theore:approx}.
Since $\|dU\|_F^2=\sum_{\alpha\in\mathcal{S}}\eta_\alpha(\kappa)$ is exponentially small, the sampling error is too.
From now on, we focus on typical circuits, emphasizing that the portion of atypical circuits is exponentially small.

Consider a special case ($\kappa=1/2$) where interference between photons from different sources is negligible typically.
In this case, for SPBS, possible outputs can be described by a disconnected graph, in which at most two vertices are connected; thus, the treewidth is 1.
For GBS, assuming that a single source emits constant number of photons at most, 
graphs describing possible outcomes are again disconnected with constant number of vertices and have bounded treewidth.
One may also show that sampling for this regime is easy by noting that the hafnian of a low-rank matrix can be efficiently computed without the assumption \cite{bjorklund2019faster}.
Thus,
\begin{theorem}(Efficient-sampling regime)\label{thm:easy}
     Approximate BS can be efficiently performed for typical circuits of depth $D\leq D_\text{easy}\equiv dk^{2/d}N^{2(\gamma-1)/d-\epsilon}/8=\Theta(N^{2(\gamma-1)/d-\epsilon})$ .
     Especially for $d=1$, the upper bound becomes $D\leq k^2\kappa^2 N^{2(\gamma-1)-\epsilon}(\log N)^2/2$.
\end{theorem}
We note that the distinct upper-bound for 1D arises because the treewidth $O(\log N)$ can be efficiently simulated.

\begin{figure}[t]
\includegraphics[width=200px]{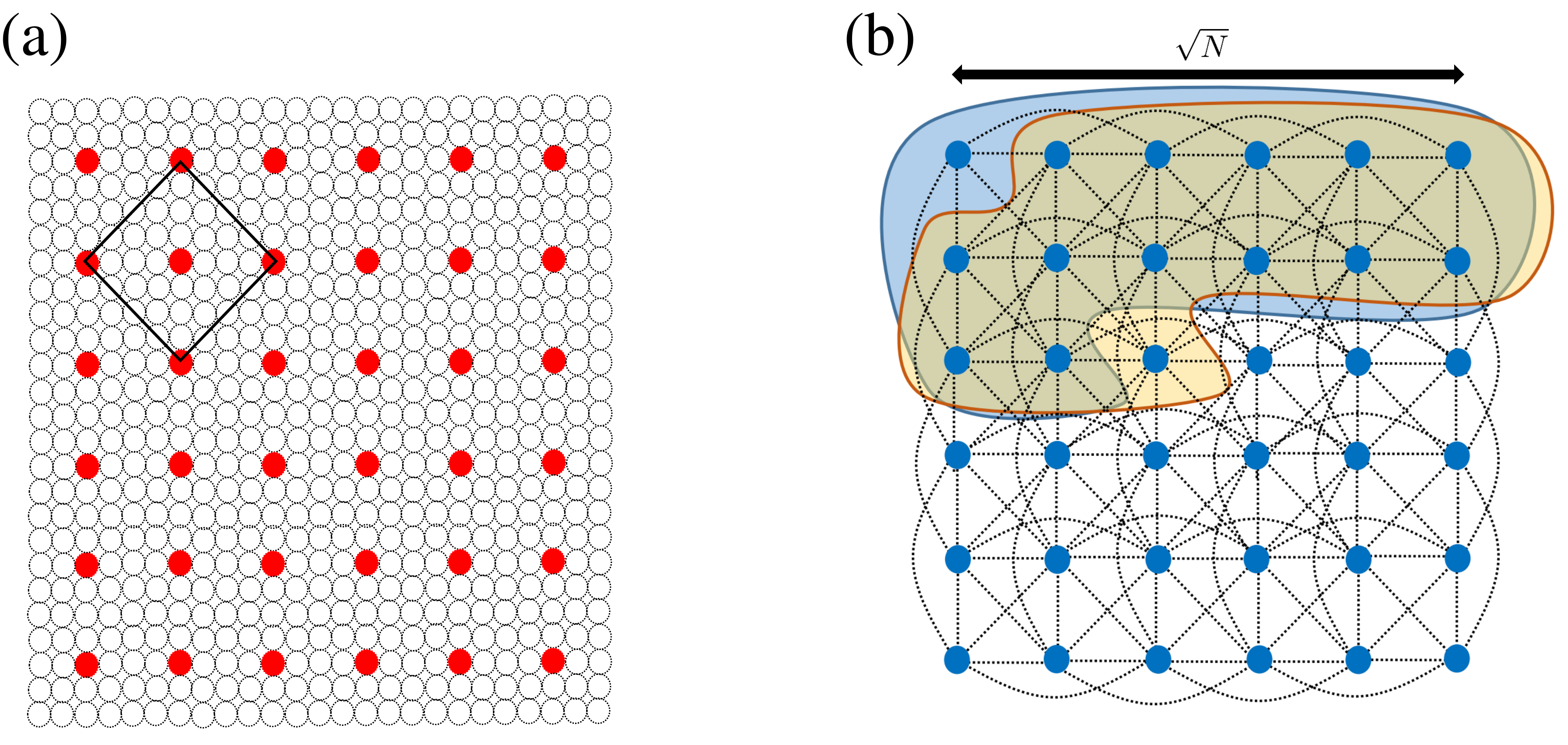}
\caption{GBS on 2D lattice with $N=36$, $M=N^2$, and $\gamma=2$. (a) Red dots represent initial sources. Black solid line describes the region at which a particular input photon can typically propagate for $D=\Theta(L^{2(1-\epsilon)})=\Theta(N^{1-\epsilon})$. (b) Possible tree decomposition of the symmetric graph $B_{\bm m}$ when outputs are at the same position with input sources. An upper bound on the treewidth is $\Theta(\sqrt{N})$ as shown: the first bag (blue) and the second one (yellow).}
\label{fig:tree}
\end{figure}

After $D>D_\text{easy}$ (or $\kappa>1/2$), photons from a sublattice can now propagate to other lattices so that photons from different sources start to interfere (see Fig.~\ref{fig:tree} (a)).
Thus, induced graphs have edges between sources and photons from different sublattices (SPBS) or photons from different sources (GBS) as shown in Fig.~\ref{fig:tree} (b).
In this case for 2D architecture, there exists an outcome corresponding to a graph containing a grid whose treewidth is unbounded, i.e., $w=\sqrt{N}$.
Therefore, the sampling complexity starts to scale exponentially in $\sqrt{N}$~\cite{supple}, which reveals a sharp transition of the complexity at $D=D_\text{easy}$ from polynomial to sub-exponential.
Similarly, when photons propagate further and for arbitrary dimension, i.e., $D=\Theta(N^{2\alpha/d} D_\text{easy}$) with $0 \leq \alpha \leq 1$ (equivalently $\kappa=\Theta(N^{\alpha/d}$)), we can find a tree decomposition whose width is $\Theta(N^{\frac{\alpha}{d}+\frac{d-1}{d}})$ for any outcomes.
Therefore, we have the following theorem:
\begin{theorem}((sub-)exponential regime)
    One can sample from typical linear-optical circuits of $D=\Theta(N^{2\alpha/d} D_\text{easy})$ with $ 0 \leq \alpha \leq 1$ by complexity $O(\poly(N)2^{\Theta(N^{\frac{\alpha}{d}+\frac{d-1}{d}})})$.
\end{theorem}
Especially when $\alpha=1$, any photons can propagate to all modes, i.e., photons fully interfere each other, which forms the complete graph for all outcomes, so that treewidth becomes $\Theta(N)$.
Since generic global Haar-random circuits are fully connected, at least $\Theta(N^{2\gamma/d})=\Theta(M^{2/d})$ order of depth is required to implement a global Haar-random circuit using a local Haar-random circuit and such an input configuration.
Fig.~\ref{fig:complexity} summarizes the result.

Interestingly, the recent GBS experiments' circuit depth scales as $\sqrt{M}$ \cite{zhong2020quantum, zhong2021phase}, which implies that their circuit is not sufficient to form a global Haar-random circuit.
Nevertheless, aside from the deviation from global Haar-random matrices, locality in their circuit is not apparent because the scale is intermediate while our analysis focuses on an asymptotic regime.
Therefore, our approximate algorithm might result in a large simulation error for this intermediate-scale GBS because of a large constant factor of the error.

One may also consider other initial configurations under local Haar-random circuits, for example sources are concentrated on a certain region.
We show that for those cases, one already needs a depth $D=\Theta(ND_\text{easy})$ to reach collision-free regime, and thus collision occurs with a high probability, while equally spaced sources reach the collision-free regime when $D=\Theta(D_\text{easy})$ \cite{supple}.



\begin{figure}[t]
\includegraphics[width=220px]{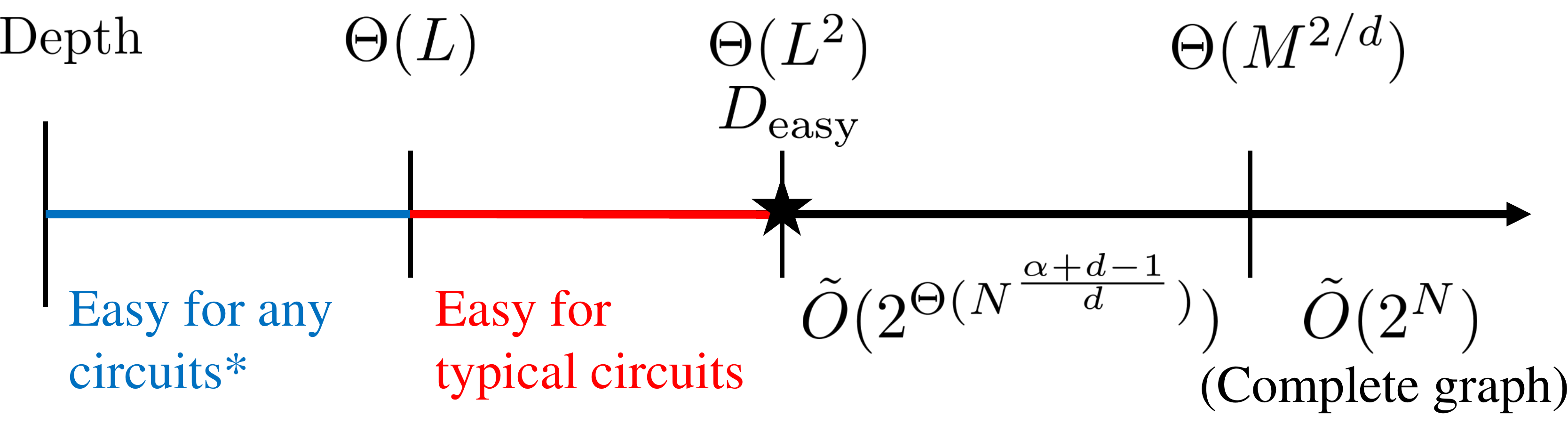}
\caption{The complexity diagram for local Haar-random BS. As the star-marked, a sharp transition occurs for the complexity of our algorithm. 
Easiness for any circuits (*) is proved in Ref.~\cite{deshpande2018dynamical}. Note that for 1D, the depth that is easy for typical circuits is larger (see Theorem 3).}
\label{fig:complexity}
\end{figure}

\emph{GBS validation test.---} \label{sec:USTC}
Finally, we implement the likelihood test to experimental samples \cite{zhong2021phase} against samples generated by our treewidth-based approximate algorithm:
\begin{align}
    \text{ratio}\equiv\log\frac{\text{Pr}_\text{ideal}(\text{Samples from experiment})}{\text{Pr}_\text{ideal}(\text{Samples from treewidth algorithm})},
\end{align}
which is equivalent to the test implemented in Refs.~\cite{zhong2020quantum, zhong2021phase}.
Thus, we compare the likelihood of each sample set with respect to the (lossy) ideal probability distribution.


\begin{figure}[t]
\includegraphics[width=220px]{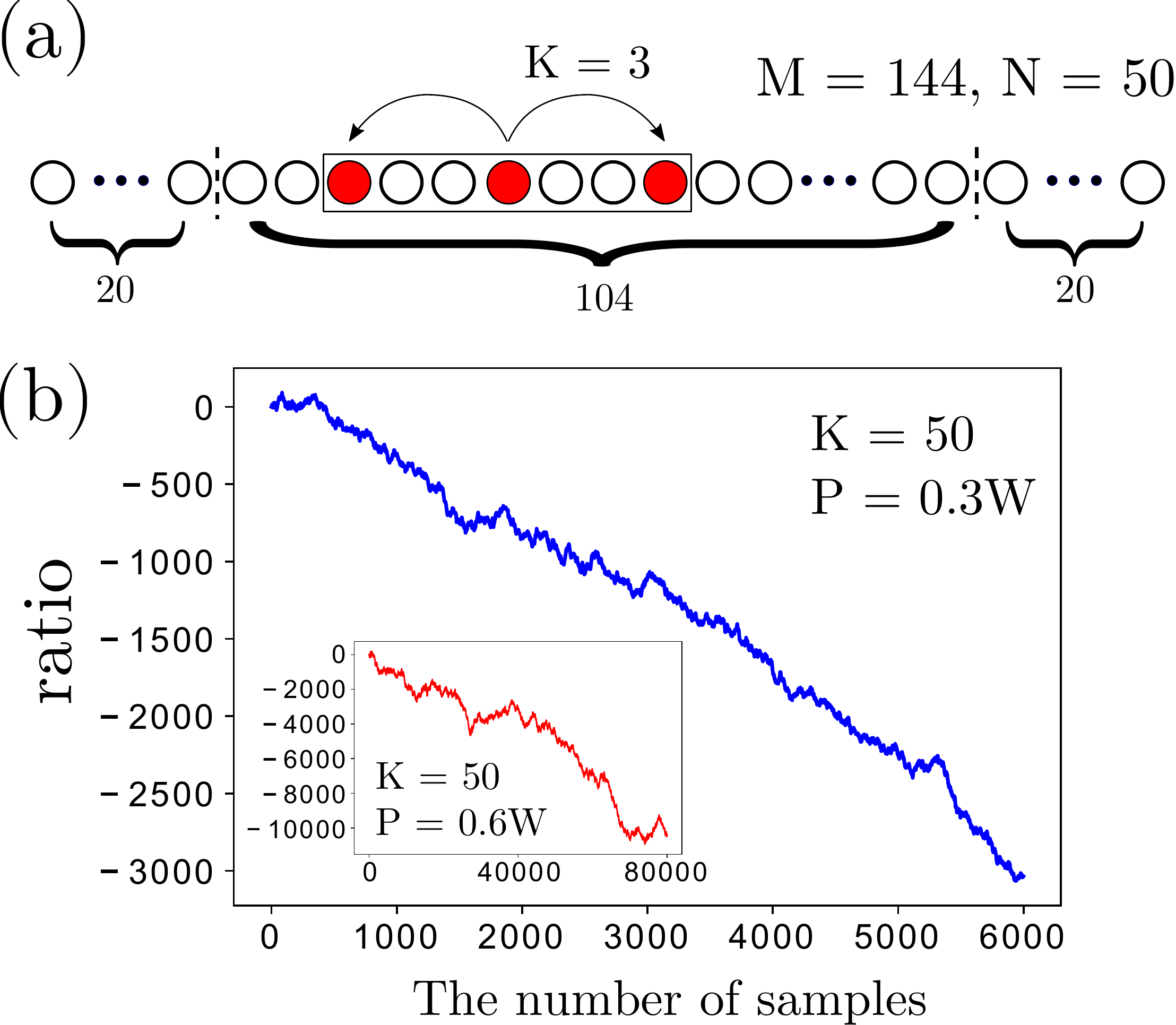}
\caption{Likelihood test for the recent GBS experiment~\cite{zhong2021phase}. (a) Rearranged mode-configuration with squeezed states sources (red dots). For approximated sampling, we discard elements of a circuit matrix $U$ that is farther than $K$ for the sources.
(b) Log-likelihood ratio of experimental samples against those from the treewidth algorithm.
}
\label{fig:HOG}
\end{figure}

For the treewidth algorithm, we have approximated local Haar-random with limited propagation (see Fig.~\ref{fig:HOG} (a)) and sampled from the approximated circuit using Theorem~\ref{theore:approx}.
Specifically, we have rearranged the 144 modes one-dimensionally and set a propagation length $K$ for approximation. Note that setting a propagation length $K$ implies that the corresponding GBS's treewidth is $w=2K+1$ and that a complete graph has a treewidth $w=M$.
To compensate the lost photons from the approximation, we have increased the squeezing parameters and thermal photons to have the same average total photon numbers.

In Fig.~\ref{fig:HOG} (b), we present the likelihood ratio as the number of samples increases for two classically verifiable instances of the experiments in Ref.~\cite{zhong2021phase}.
It clearly shows that the treewidth-based approximate algorithm renders larger likelihood than the experiment.
We also provide evidence in Ref.~\cite{supple} for GBS experiments in the quantum supremacy regime by investigating the likelihood ratio for marginals that the treewidth-based algorithm might give a larger likelihood with a limited treewidth.
Therefore, the numerical results imply that a fully connected circuit is crucial for more rigorous quantum-advantage demonstration.

\emph{Discussion.---} \label{sec:previous}
We have presented classical samplers taking advantage of limited connectivity of a circuit.
It is an interesting open question to find more efficient sampling algorithms than the one based on the treewidth. 
Another open problem is to close the gap of complexity for computing loop hafnian between the treewidth-based algorithm ($2^N$) and the best-known algorithm ($2^{N/2}$)~\cite{quesada2022quadratic}. 


Finally, Theorem~\ref{thm:easy} shows that typical linear-optical circuits up to depth $D\leq D_\text{easy}=\Theta(N^{\frac{2}{d}(\gamma-1)-\epsilon})$ allow an efficient classical simulation except for an exponentially small fraction of random circuits. 
Meanwhile, there exists a circuit hard to classically simulate for $D=\Omega(N^{\frac{\gamma-1}{d}+\epsilon})$ under reasonable complexity-theoretic conjectures \cite{aaronson2011computational, deshpande2018dynamical, maskara2019complexity}.
Theorem~\ref{thm:easy} can be compatible with the hardness results since together the implication is that the worst-case instances occupy only at most an exponentially small faction of the space of all linear optical circuits.

\begin{acknowledgements}
We thank Owen Howell, Alireza Seif, Roozbeh Bassirian, Abhinav Deshpande for interesting and fruitful discussions.
C.O. and L.J. acknowledge support from the ARL-CDQI (W911NF-15-2-0067), ARO (W911NF-18-1-0020, W911NF-18-1-0212), ARO MURI (W911NF-16-1-0349), AFOSR MURI (FA9550-15-1-0015, FA9550-19-1-0399), DOE (DE-SC0019406), NSF (EFMA-1640959, OMA-1936118), and the Packard Foundation (2013-39273). Y. L. acknowledges National Research Foundation of Korea a grant funded by the Ministry of Science and ICT (NRF-2020M3E4A1077861) and KIAS Individual Grant
(CG073301) at Korea Institute for Advanced Study.
B.F. acknowledges support from AFOSR (YIP number FA9550-18-1-0148 and
FA9550-21-1-0008). 
This material is based upon work partially
supported by the National Science Foundation under Grant CCF-2044923
(CAREER).
We also acknowledge the University of Chicago’s Research Computing Center for their support of this work.
We acknowledge The Walrus python library for the open source of Gaussian boson sampling algorithms \cite{gupt2019walrus}
\end{acknowledgements}

\bibliography{reference}

\begin{thebibliography}{45}
\bibitem [{loo()}]{loophaf}%
  \BibitemOpen
  \href@noop {} {\bibinfo {title} {A matrix for which we compute hafnian or
  loop hafnian is implicitly assumed to be symmetric.}}\BibitemShut {Stop}%
\bibitem [{ban()}]{bandwidth}%
  \BibitemOpen
  \href@noop {} {\bibinfo {title} {For a given matrix $a$, the bandwidth is the
  smallest $k$ such that $a_{ij}=0$ for all $i, j$ such that
  $|i-j|>k$.}}\BibitemShut {Stop}%
\bibitem [{equ()}]{equal}%
  \BibitemOpen
  \href@noop {} {\bibinfo {title} {Our main results are valid even if initial
  sources are not equally distributed but separated by a minimum distance of
  ${\Theta}({L})$ from each other.}}\BibitemShut {Stop}%
\end{thebibliography}

\end{document}


\preprint{APS/123-QED}
\title{Classical simulation of boson sampling based on graph structure : Supplemental Material}
\author{Changhun Oh}%
\affiliation{Pritzker School of Molecular Engineering, University of Chicago, Chicago, Illinois 60637, USA}
\author{Youngrong Lim}%
\affiliation{School of Computational Sciences, Korea Institute for Advanced Study, Seoul 02455, Korea}
\author{Bill Fefferman}
\affiliation{Department of Computer Science, University of Chicago, Chicago, Illinois 60637, USA}
\author{Liang Jiang}
\affiliation{Pritzker School of Molecular Engineering, University of Chicago, Chicago, Illinois 60637, USA}
\date{\today}

\maketitle

\setcounter{equation}{0}
\setcounter{lemma}{2}
\renewcommand{\thesection}{S\arabic{section}}
\renewcommand{\theequation}{S\arabic{equation}}
\renewcommand{\thefigure}{S\arabic{figure}}

\tableofcontents

\newpage
\section{Dynamical Programming Algorithm (Proof of Lemma 1)}
In this section, we prove Lemma 1 by explicitly providing algorithms to compute permanent or loop hafnian of an $N\times N$ matrix with the complexity $O(Nw^22^w).$

Before providing our algorithms, we note that methods using tree decomposition and treewidth have broad applications to many areas in computer science and applied mathematics~\cite{dechter2003constraint,bodlaender2008combinatorial}. There are many computational problems and corresponding algorithms whose complexity depends on the treewidth instead of the system size.
Remarkably, algorithms based on the treewidth are believed to be almost optimal for specific problems in that only a tiny improvement of the tree algorithm makes fail the Exponential Time Hypothesis~\cite{marx2007can}.

We also note that in general, the exact computation of the treewidth for an arbitrary graph is NP-hard problem~\cite{arnborg1987complexity}. 
Nevertheless, we still have good approximation algorithms of treewidth~\cite{bodlaender2008combinatorial}, and for fixed $k$, a linear-time algorithm whether the treewidth at most $k$ and finding the corresponding tree decomposition~\cite{bodlaender1996linear}. 
In addition, a reasonable upper bound on the treewidth is enough in practical situations.

\subsection{Dynamical Programming Algorithm (permanent)}\label{permanent}

\subsubsection{Tree decomposition of a bipartite graph}
We generalize the algorithm computing permanent by using tree decomposition proposed in Ref.~\cite{cifuentes2016efficient} to apply collision cases in single-photon boson sampling.
Let us first recall the algorithm of Ref.~\cite{cifuentes2016efficient}.
Consider an $M\times M$ matrix $U$.
Let $A$ and $X$ denote the subset of its rows and columns, respectively, i.e., $A=\{\alpha_{1},\dots,\alpha_k\}$ and $X=\{r_1,\dots,r_k\}$, where $\alpha_i$'s and $r_i$'s are in $[M]\equiv \{1,\dots,M\}$.
For collision-free cases, redundant elements do not appear in $A$ and $X$, while we can apply the algorithm even when there is redundancy by extending matrix $U$ by repeating the corresponding rows and columns and treating them as distinct row and column elements.
However, since the algorithm proposed in Sec.~\ref{sec:per_collision} is more efficient than the latter for collision cases and for simplicity, we assume that they are distinct.
Now, we define a bipartite graph $G(U,A,X)$ for the matrix $U$ as a graph having vertices $A\cup X$ and edges $(a,x)$ if $U_{a,x}\neq 0$.

We now introduce a tree decomposition $(T,\chi)$ of a bipartite graph $G$.
Tree decomposition of a graph $G$ needs to satisfy the following conditions:

i-1) The union of $\{\alpha(t)\}_{t\in T}$ is the whole row set $A$.

i-2) The union of $\{\chi(t)\}_{t\in T}$ is the whole column set $X$.

ii) For every edge $(a,x)$ of $G$, there exists a node $t$ of $T$ with $a\in\alpha(t)$, $x\in \chi(t)$.

iii-1) For every $a\in A$, the set $\{t:a\in\alpha(t)\}$ forms a subtree of $T$.

iii-2) For every $x\in X$, the set $\{t:x\in\chi(t)\}$ forms a subtree of $T$.

The width $w$ of a tree decomposition is the largest of $|\alpha(t)|+|\chi(t)|-1$ among all nodes $t$, and the treewidth is the smallest width $w$ over all possible tree decompositions for a given graph.
Here, the subtree rooted in $t$ is the tree that consists of $t$ and its descendants.

We now define a partial permanent for some sets $D\subset A$ and $Y\subset X$ as
\begin{align}
    \per(D,Y)\equiv \sum_{\pi\in \BPM(G)} \prod_{a\in D}U_{a,\pi(a)},
\end{align}
where the sum is over all bipartite perfect matchings $\pi : D\to Y$.
If $|D|\neq |Y|$, then $\per(D,Y)=0$ because they cannot constitute a perfect matching.
Here, the goal is to compute $\Per(U^{\alpha_1,\dots,\alpha_k}_{r_1,\dots,r_k})=\per(A,X)$ (see Sec. \ref{SM:clifford}) using the tree decomposition and analyze the complexity.

\subsubsection{Collision-free cases}
Now we present the recursion formula to compute $\Per(U^{\alpha_1,\dots,\alpha_k}_{r_1,\dots,r_k})$ \cite{cifuentes2016efficient}:
\begin{lemma}\label{lemma:permanent}
    Let $U$ be a matrix with associated graph $G$.
    Let $(T,\alpha,\chi)$ be a tree decomposition of $G$.
    Let $t$ be an internal node of $T$, and let $D$ and $Y$ be such that
    \begin{align}
        \alpha(T_{t})\setminus \alpha(t)\subset D \subset \alpha(T_{t}), ~~~\chi(T_{t})\setminus \chi(t)\subset Y \subset \chi(T_{t}), ~~~ |D|=|Y|.
    \end{align}
    
    Then,
    \begin{align}\label{eq:perm_res}
        \per(D,Y)
        =\sum_{D_t,D_{c_j},Y_t,Y_{c_j}}\per(D_{t},Y_{t})\prod_{j=1}^l\per^*(D_{c_{j}},Y_{c_{j}}),
    \end{align}
    where the sum is over all $\mathcal{D}=(D_t,D_{c_1},\dots), \mathcal{Y}=(Y_t,Y_{c_1},\dots)$ such that:
    \begin{align}\label{eq:perm_cond}
    &D=D_{t}\sqcup (D_{c_{1}}\sqcup \dots \sqcup D_{c_l}), ~~~
    Y=Y_{t}\sqcup (Y_{c_{1}}\sqcup \dots \sqcup Y_{c_l}), \\ 
    &\alpha(T_{c_{j}})\setminus \alpha(t)\subset D_{c_{j}}\subset \alpha(T_{c_{j}}),~~~
    D_{t}\subset \alpha(t), \\ 
    &\chi(T_{c_{j}})\setminus \chi(t)\subset Y_{c_{j}}\subset \chi(T_{c_{j}}),~~~
    Y_{t}\subset \chi(t).
\end{align}
\end{lemma}
Here, $l$ is the number of children of the node $t$, and we define
\begin{align}
    \per^*(D,Y)\equiv \sum_\pi \prod_{a\in D} U_{a,\pi(a)},
\end{align}
where the sum is over all bijections $\pi:D\to Y$ such that $\pi\cap (\alpha(t)\times \chi(t))=\emptyset$.

We provide a proof sketch and more detailed proof can be found in Ref. \cite{cifuentes2016efficient}.
Let us consider a matching $\pi:D\to Y$ and decompose elements as
\begin{align}
    D=(D\cap \alpha(t))\sqcup (D\cap \alpha(T_{c_1})\setminus \alpha(t))\sqcup\dots\sqcup(D\cap \alpha(T_{c_l})\setminus \alpha(t)), \\ 
    Y=(Y\cap \chi(t))\sqcup (Y\cap \chi(T_{c_1})\setminus \chi(t))\sqcup\dots\sqcup(Y\cap \chi(T_{c_l})\setminus \chi(t)).
\end{align}
Using definition of tree decomposition, one can find that the matching can be decomposed as
\begin{align}\label{eq:perm_decom}
    \pi=\pi_t\sqcup \pi_{c_1}\sqcup \cdots \sqcup \pi_{c_l}
\end{align}
in such a way that $\pi_t\subset(D\cap \alpha(t))\times (Y\cap \chi(t))$ and for each child $c$ we have that $\pi_c\subset(D\cap \alpha(T_c))\times(Y\cap\chi(T_c))$, where $\pi_{c_j}\cap (\alpha(t)\times \chi(t))=\emptyset$.
Also, the decomposition is unique.
By defining $Y_t$ and $Y_{c_i}$ to be the range of $\pi_t$ and $\pi_{c_i}$, respectively, and $D_t$ and $D_{c_i}$ to be the domains of $\pi_t$ and $\pi_{c_i}$, they satisfy Eq. \eqref{eq:perm_cond}.
Therefore, matchings between $D$ and $Y$ can be decomposed as Eq.~\eqref{eq:perm_decom}, which leads to Eq.~\eqref{eq:perm_res}.

Here, by using inclusion-exclusion principle, one can find that \cite{cifuentes2016efficient}
\begin{align}\label{eq:per_star}
    \per^*(D_c,Y_c)=\sum_{\substack{D_{tc}\sqcup D_{cc}=D_c,\\ Y_{tc}\sqcup Y_{cc}=Y_c,\\ D_{tc}\subset\alpha(t),Y_{tc}\subset\chi(t)}}(-1)^{|D_{tc}|}\per(D_{tc},Y_{tc})\per(D_{cc},Y_{cc}),
\end{align}
where $|D_{tc}|$ is the number of elements in $D_{tc}$.
More explicitly, observe that for $D_{tc}\subset D_t\cap \alpha(t)$ and $Y_{tc}\subset Y_t\cap\chi(t)$,
\begin{align}
    \per(D_{tc},Y_{tc})\per(D_c\setminus D_{tc},Y_c\setminus Y_{tc})
    =\sum_{\pi_c^*(D_{tc},Y_{tc})}\prod_{a\in D_{c}} U_{a,\pi^*_c(a)},
\end{align}
where the sum is over perfect matchings between $D_c$ and $Y_c$, $\pi_c^*$ that contains a perfect matching between $D_{tc}$ and $Y_{tc}$.
Recall that $\per^*(D_c,Y_c)$ is the sum over matchings that do not contain any matching between $\alpha(t)$ and $\chi(t)$, i.e., $\per^*(D_c,Y_c)=\sum_{\pi_c^*(\emptyset,\emptyset)}\prod_{a\in D_{tc}} U_{a,\pi^*_c(a)}$
Thus, inclusion-exclusion principle leads to Eq.~\eqref{eq:per_star}.
Therefore, one can equivalently rewrite Eq.~\eqref{eq:perm_res} as
\begin{align}
    \per(D,Y)=\sum_{\mathcal{D},\mathcal{Y}}\per(D_t,Y_t)\prod_{j=1}^l(-1)^{|D_{tc_j}|}\per(D_{tc_j},Y_{tc_j})\per(D_{cc_j},Y_{cc_j}),
\end{align}
where the sum is over $\mathcal{D}=(D_t,D_{tc_1},D_{cc_1},\dots)$ and $\mathcal{Y}=(Y_t,Y_{tc_1},Y_{cc_1},\dots)$ satisfying
\begin{align}
    &D=D_{t}\sqcup (D_{tc_{1}}\sqcup D_{cc_{1}}\sqcup \dots \sqcup D_{tc_l}\sqcup D_{cc_l}), ~~~
    Y=Y_{t}\sqcup (Y_{tc_{1}}\sqcup Y_{cc_{1}}\sqcup\dots \sqcup Y_{tc_l}\sqcup Y_{cc_l}), \\ 
    &\alpha(T_{c_{j}})\setminus \alpha(t)\subset D_{cc_{j}}\subset \alpha(T_{c_{j}}),~~~
    D_{t}\subset \alpha(t), ~~~ D_{tc_j}\subset \alpha(t)\cap \alpha(t_{c_j}),\\ 
    &\chi(T_{c_{j}})\setminus \chi(t)\subset Y_{c_{j}}\subset \chi(T_{c_{j}}),~~~ 
    Y_{t}\subset \chi(t),~~~
    Y_{tc_j}\subset \chi(t)\cap \chi(t_{c_j}).
\end{align}
It proves Lemma~\ref{lemma:permanent}.
Using the recursive relation, we finally obtain the quantity of $\per(A,X)$.

Now, let us analyze the complexity.
We first need to compute $\per(D_t,Y_t)$ for all node $t$'s.
For a node $t$, since $D_t\subset \alpha(t)$ and $Y_t\subset \chi(t)$, one can evaluate all $\per(A, X)$ with $A\subset D$ and $X\subset Y$ in $O(w^22^w)$ \cite{cifuentes2016efficient}:
Consider an $n_1\times n_2$ matrix $U$.
Let $A_0$ and $X_0$ be its row and column set, respectively, and $S=\{(D,Y)\subset A_0\times X_0:|D|=|Y|\}$.
Then, one can compute $\per(D,Y)$ for all $(D,Y)\in S$ recursively as follows.
Consider $(D,Y)$ with $|D|=|Y|=i$.
Then, letting $a_0$ be the first element in $D$, we can compute its permanent by using $(D',Y')$ with $|D'|=|Y'|=i-1$,
\begin{align}
    \per(D,Y)=\sum_{x\in Y}U_{a_0,x}\per(D\setminus a_0,Y\setminus x).
\end{align}
Since the loop is over at most $n_2$ elements and finding the sets $D\setminus a_0$ and $Y\setminus x$ costs at most $O(n_1+n_2)$,
Thus, its complexity is $O((n_1+n_2)^22^{n_1+n_2})$.

Once we know all the values of $\per(D_t,Y_t)$, $\per(D_{tc_j},Y_{tc_j})$, and $\per(D_{cc_j},Y_{cc_j})$, we can then evaluate $\per(D,Y)$ for all $D$ and $Y$ in $O(lw^22^w)$ by using a subset convolution over the subsets of $\alpha(t)\times \chi(t)$ \cite{bjorklund2007fourier}.
Therefore, the complexity of computing $\Per(U^{\alpha_1,\dots,\alpha_k}_{r_1,\dots,r_k})$ is $O(k w^22^w)$ since the sum of the number of children of all nodes is $O(k)$.

\subsubsection{Collision events}\label{sec:per_collision}
In general, more than a single photon can be detected in the same mode, i.e., $r_i$'s are not necessarily distinct.
In this case, although we may still employ the algorithm in the previous section, we can further develop the algorithm to be more efficient.
We first set a threshold of the total number of photons in a mode $c$, i.e., the maximum number of photons in an output mode is $c$, and we declare the overload if we have more than $c$ photons.
We note that when $M=\omega(N^2)$, the probability of collision is highly suppressed, i.e., $c=1$ is sufficient \cite{aaronson2011computational}.
To take into account collisions properly, we modify the definition of graphs by additionally assigning photon numbers at each vertices.

Now the goal is to compute $\Per(U^{\alpha_1,\dots,\alpha_k}_{r_1,\dots,r_k})$, where $\alpha_i$'s and $r_i$'s may not be distinct.
While we focus on single-photon boson sampling, i.e. $\alpha_i$'s are distinct, for generality, we assume that $\alpha_i$'s may not be distinct as well.
Again, let $A=\{\alpha_1,\dots,\alpha_k\}$ and $X=\{r_1,\dots,r_k\}$ with treating $\alpha_i$'s and $r_i$'s are distinct, and we consider bipartite matchings between $A$ and $X$.
We introduce weight vectors $\bm{n}$ and $\bm{m}$ such that $n_i$ is the number of repetition of $i$ in $A$ and $m_i$ is the number of repetition of $i$ in $X$ for $i\in[M]$ for later usage.
For boson sampling, $\bm{n}$ and $\bm{m}$ may correspond to an input photon and output photon configuration (of marginals), respectively (see Sec.~\ref{SM:clifford}).
Therefore, $n_i\in\{0,1\}$ for single-photon boson sampling.
Let $G_{\bm{n},\bm{m}}=G(U,A,X)$ be the underlying graph structure of $U$, with row and column vertex sets $A$ and $X$.

Meanwhile, for $A$ and $X$, we define $\bar{A}$ and $\bar{X}$ to be the sets without redundancy of $A$ and $X$, respectively.
We now define some generalized bipartite perfect matchings of $G$ that allow multiple edges between vertices in $\bar{A}$ and $\bar{X}$.
We represent a list of repeated edges as a vector $\tau$, i.e., $\tau$ is a vector indexed by $E=\bar{A}\times \bar{X}$ and for each $ij\in E$ the entry $\tau_{ij}$ indicates the number of times that edge $ij$ appears.
The degree vector of $a\in \bar{A}$ is the vector $\deg_{\bar{A}}(a)$ with coordinates $\deg_{\bar{A}}(a)=\sum_{x\in \bar{X}}\tau_{a x}$, and the degree vector of $x\in \bar{X}$ is the vector $\deg_{\bar{X}}(x)$ with coordinates $\deg_{\bar{X}}(x)=\sum_{a\in \bar{A}}\tau_{a x}$.
For weight vectors $\bm{n}$ and $\bm{m}$, we define an $(\bm{n},\bm{m})$-matching of $G_{\bm{n},\bm{m}}$ as a vector $\tau$ such that $\deg_{\bar{A}}(a_i)=n_i$ and $\deg_{\bar{X}}(x_i)=m_i$ for all $a_i\in \bar{A}$ and $x_i\in \bar{X}$.
Let $\text{BPM}_G(\bm{n},\bm{m})$ be the set of all $(\bm{n},\bm{m})$-matchings of $G$.

Let $G_{\bm{n},\bm{m}}$ be the graph associated to $U$.
We may view its row vertices as $(i,l_i)$, where $i\in \bar{A}$, $l_i\in [n_i]$, and column vertices as pairs $(j,k_j)$, where $j\in \bar{X}$, and $k_j\in[m_j]$.
The edges are between the vertices.
Given a bipartite perfect matching $\pi\in\text{BPM}(G_{\bm{n},\bm{m}})$, there is a natural way to obtain an $(\bm{n},\bm{m})$-matching $\tau\in\text{BPM}_G(\bm{n},\bm{m})$.
Namely, for each edge $[(i,l_i), (j,k_j)]$ in $\pi$ we drop the second coordinate and obtain the edge $ij$. 
This gives a function $f:\text{BPM}(G_{\bm{n},\bm{m}})\to \text{BPM}_G(\bm{n},\bm{m})$.
Thus, one can easily show that the fiber $f^{-1}$ consists of $\bm{m}!\bm{n}!/\bm{\tau}!$ elements.
Then,
\begin{align}
    \Per(U^{\alpha_1,\dots,\alpha_k}_{r_1,\dots,r_k})
    \equiv \sum_{\pi\in\text{BPM}(G_{\bm{n},\bm{m}})}\prod_{a\in A}U_{a,\pi(a)} =\sum_{\tau\in\text{BPM}_G(\bm{n},\bm{m})}\frac{\bm{m}!\bm{n}!}{\bm{\tau}!}\prod_{a\in \bar{A}}\prod_{\tau(a)} U_{a,\tau(a)}.
\end{align}
Note that the product over $\tau(a)$ appears because $a\in \bar{A}$ can be connected multiple vertices on $X$.

From now on, we assume that the weight vectors $\bm{n}$ and $\bm{m}$ are fixed.
Given another weight vectors $\bm{d}\leq \bm{n}$ and $\bm{y}\leq \bm{m}$ (component-wise), the associated scaled partial permanent is defined as
\begin{align}
    \per(\bm{d}, \bm{y})\equiv \frac{1}{\bm{d}!\bm{y}!}\Per(U^{\bm{d}}_{\bm{y}})=\sum_{\tau\in\text{BPM}_G(\bm{d},\bm{y})}\frac{1}{\bm{\tau}!}\prod_{a\in \supp(\bm{d})}\prod_{\tau(a)} U_{a,\tau(a)},
\end{align}
where 
$\supp(\bm{d})\equiv\{i\in[M]:d_i\neq 0\}$.
Here, $U^{\bm{d}}_{\bm{y}}\equiv U^{\alpha_1,\dots,\alpha_k}_{r_1,\dots,r_k}$ such that the weight vectors of $\alpha_i$'s and $r_i's$ are equal to $\bm{d}$ and $\bm{y}$, respectively.
Then, the permanent that we want to compute is written as
\begin{align}
    \Per(U^{\bm{n}}_{\bm{m}})=\bm{n}!\bm{m}!\per(\bm{n},\bm{m}).
\end{align}

Let us also define $\sat(\bm{d})=\{i\in[M]:d_i=n_i\}$ for a given $\bm{n}$, and $\sat(\bm{y})=\{i\in[M]:y_i=m_i\}$ and $\supp(\bm{y})=\{i\in[M]:y_i\neq 0\}$ for a given $\bm{m}$.
With a tree decomposition $(T,\alpha,\chi)$ for $G_{\bm{n},\bm{m}}$,
we now provide the recursive relation to compute $\Per(U^{\bm{n}}_{\bm{m}})$.

\begin{lemma}
    Let $U$ be a matrix with associated graph $G$.
    Let $(T,\alpha,\chi)$ be a tree decomposition of $G$.
    Let $t$ be an internal node of $T$, and let $\bm{d}$ and $\bm{y}$ be such that
    \begin{align}\label{eq:condition_per}
        \alpha(T_t)\setminus \alpha(t)\subset \sat(\bm{d})\subset \supp(\bm{d})\subset \alpha(T_t), ~~~ 
        \chi(T_t)\setminus \chi(t)\subset \sat(\bm{y})\subset \supp(\bm{y})\subset \chi(T_t),~~~ |\bm{d}|=|\bm{y}|,
    \end{align}
    where $|\bm{d}|=\sum_{i=1}^M d_i$, and $|\bm{y}|$ is similar.
    Then,
    \begin{align}\label{eq:perm_collision_res}
        \per(\bm{d},\bm{y})=\sum_{\bm{d}_t,\bm{d}_c,\bm{y}_t,\bm{y}_c}\per(\bm{d}_t,\bm{y}_t)\prod_{j=1}^l\per^*(\bm{d}_{c_{j}},\bm{y}_{c_j}),
    \end{align}
    where the sum is over $(\bm{d}_t,\bm{d}_{c_1},\dots,\bm{d}_{c_l})$ and $(\bm{y}_t,\bm{y}_{c_1},\dots,\bm{y}_{c_l})$ such that:
    \begin{align}\label{eq:per1}
        &\bm{d}=\bm{d}_t+\sum_{j=1}^l \bm{d}_{c_j}, ~~~ 
        \bm{y}=\bm{y}_t+\sum_{j=1}^l \bm{y}_{c_j},\\ 
        &\alpha(T_{c_j})\setminus \alpha(t)\subset\sat(\bm{d}_{c_j})\subset \supp(\bm{d}_{c_j})\subset \alpha(T_{c_j}), ~~~ \supp(\bm{d}_t)\subset \alpha(t),\label{eq:per2}\\ 
        &\chi(T_{c_j})\setminus \chi(t)\subset\sat(\bm{y}_{c_j})\subset \supp(\bm{y}_{c_j})\subset \chi(T_{c_j}),  ~~~ \supp(\bm{y}_t)\subset \chi(t),\label{eq:per3} \\
        &\supp(\bm{d}_t)\sqcup\supp(\bm{d}_{c_1})\sqcup \dots \sqcup \supp(\bm{d}_{c_l})=\supp(\bm{d}),\label{eq:per4} \\ 
        &\supp(\bm{y}_t)\sqcup\supp(\bm{y}_{c_1})\sqcup \dots \sqcup \supp(\bm{y}_{c_l})=\supp(\bm{y}).\label{eq:per5}
    \end{align}
    Here, $\per^*(\bm{d}_c,\bm{y}_c)$ is defined as
    \begin{align}
        \per^*(\bm{d}_c,\bm{y}_c)\equiv \sum_{\tau\in \BPM_G(\bm{d}_c,\bm{y}_c)}\frac{1}{\bm{\tau}!}\prod_{a\in\supp(\bm{d}_c)}\prod_{\tau(a)}U_{a,\tau(a)},
    \end{align}
    where $\tau\cap (\alpha(t)\times\chi(t))=\emptyset$.
\end{lemma}
\begin{proof}
    Let us consider $\bm{d}$ and $\bm{y}$ satisfying Eq.~\eqref{eq:condition_per} and a generalized matching $\tau\in\BPM_G(\bm{d},\bm{y})$.
    Here, $\tau \subset D\times Y$ with $D\equiv \supp(\bm{d})$ and $Y\equiv \supp(\bm{y})$.
    Consider the following decompositions:
    \begin{align}
        D=(D\cap \alpha(t))\sqcup(D\cap\alpha(T_{c_1})\setminus \alpha(t))\sqcup \cdots \sqcup (D\cap\alpha(T_{c_l})\setminus \alpha(t)), \\ 
        Y=(Y\cap \chi(t))\sqcup(Y\cap\chi(T_{c_1})\setminus \chi(t))\sqcup \cdots \sqcup (Y\cap\chi(T_{c_l})\setminus \chi(t)).
    \end{align}
    Let $\tau_t$ be the submatching of $\tau$ with domain contained in $D\cap\alpha(t)$ and range contained in $Y\cap \chi(t)$.
    If some $a\in D\cap \alpha(t)$ is not in the domain of $\tau_t$, there is a child $c$ such that $a\in \alpha(T_c)$ and $\pi(a)\subset\chi(T_c)$ by definition of tree decomposition.
    Also, if $x\in Y\cap \chi(t)$ is not in the range of $\tau_t$, then there is a child $c$ such that $x\in\chi(T_c)$ and $\tau^{-1}(x)\subset\alpha(T_c)$.
    Similarly, we have
    \begin{align}
        \tau(D\cap \alpha(T_c)&\setminus \alpha(t))\subset Y \cap \chi(T_c) \\ \tau(Y\cap \alpha(T_c)&\setminus \chi(t))\subset D \cap \alpha(T_c).
    \end{align}
    Hence, we can decompose
    \begin{align}
        \tau=\tau_t\sqcup \tau_{c_1}\sqcup\cdots\sqcup \tau_{c_l},
    \end{align}
    in such a way that $\tau_t\subset (D\cap\alpha(t))\times(Y\cap \chi(t))$ and for each child $c$ we have that $\tau_c\subset(D\cap \alpha(T_c))\times(Y\cap \chi(T_c))$ and $\tau_c\cap (\alpha(t)\times\chi(t))=\emptyset$.
    One can easily check that this decomposition is unique.
    Now, we decompose given weight vectors $\bm{d}$ and $\bm{y}$ following the decomposition, which leads to the conditions \eqref{eq:per1}-\eqref{eq:per5}.
    Thus for a given generalized matching $\tau$ we decompose it, which leads to Eq.~\eqref{eq:perm_collision_res}.
    Also, when submatchings satisfying Eqs.~\eqref{eq:per1}-\eqref{eq:per5} are given, we can combine them to constitute a matching $\tau$.
\end{proof}

Again, we show that
\begin{align}\label{eq:collision_star}
    \per^*(\bm{d}_c,\bm{y}_c)=\sum_{\bm{d}_{tc},\bm{d}_{cc},\bm{y}_{tc},\bm{y}_{cc}} (-1)^{|\bm{d}_{tc}|}\per(\bm{d}_{tc},\bm{y}_{tc})\per(\bm{d}_{cc},\bm{y}_{cc}),
\end{align}
where the sum is over $\bm{d}_{tc}$, $\bm{d}_{cc}$, $\bm{y}_{tc}$, $\bm{y}_{cc}$ satisfying
\begin{align}
    &\bm{d}_{tc}+\bm{d}_{cc}=\bm{d}_c,~~~ \bm{y}_{tc}+\bm{y}_{cc}=\bm{y}_c,~~~ \supp(\bm{d}_{tc})\subset \alpha(t), ~~~ \supp(\bm{y}_{tc})\subset \chi(t) \\
    &\supp(\bm{d}_c)=\supp(\bm{d}_{cc})\sqcup \supp(\bm{d}_{tc}),~~~
    \supp(\bm{y}_c)=\supp(\bm{y}_{cc})\sqcup \supp(\bm{y}_{tc}).
\end{align}
To show that, observe that for some $\bm{d}_{tc}\leq\bm{d}_c$ and $\bm{y}_{tc}\leq\bm{y}_c$ (elementwise) with $\supp(\bm{d}_{tc})\subset \alpha(t)$ and $\supp(\bm{y}_{tc})\subset \chi(t)$ and $\supp(\bm{d}_c)=\supp(\bm{d}_{tc})\sqcup \supp(\bm{d}_{c}-\bm{d}_{tc}), \supp(\bm{y}_c)=\supp(\bm{y}_{c})\sqcup \supp(\bm{y}_{c}-\bm{y}_{tc})$,
\begin{align}
    \per(\bm{d}_{tc},\bm{y}_{tc})\per(\bm{d}_c-\bm{d}_{tc},\bm{y}_c-\bm{y}_{tc})
    =\sum_{\tau_c(\bm{d}_{tc},\bm{y}_{tc})\in \BPM(\bm{d}_c,\bm{y}_c)}\frac{1}{\bm{\tau}_c!}\prod_{a\in\supp(\bm{d}_c)}\prod_{\tau(a)}U_{a,\tau_c(a)},
\end{align}
where the sum is over $\tau_c$ that contains a generalized perfect matching between $\bm{d}_{tc}$ and $\bm{y}_{tc}$.
Since $\per^*(\bm{d}_c,\bm{y}_c)$ corresponds to the sum over generalized perfect matchings that do not contain any matching between $\alpha(t)$ and $\chi(t)$, by using inclusion-exclusion principle, we obtain the expression of~\eqref{eq:collision_star}.
Therefore, we can rewrite Eq.~\eqref{eq:perm_collision_res} as
\begin{align}
    \per(\bm{d},\bm{y})=\sum\per(\bm{d}_t,\bm{y}_t)\prod_{j=1}^l(-1)^{|\bm{d}_{tc_j}|}\per(\bm{d}_{tc_j},\bm{y}_{tc_j})\per(\bm{d}_{cc_j},\bm{y}_{cc_j})
\end{align}
where the sum is over $\bm{d}_t,\bm{y}_t,\bm{d}_{tc_j},\bm{y}_{tc_j},\bm{d}_{cc_j},\bm{y}_{cc_j}$ such that
\begin{align}
    &\bm{d}=\bm{d}_t+\sum_{j=1}^l(\bm{d}_{tc_j}+\bm{d}_{cc_j}), ~~~
    \supp(\bm{d}_t)\subset \alpha(t),~~~
    \bm{y}=\bm{y}_t+\sum_{j=1}^l(\bm{y}_{tc_j}+\bm{y}_{cc_j}), ~~~
    \supp(\bm{y}_t)\subset \chi(t),\\
    &D=\supp(\bm{d}_t)\sqcup\supp(\bm{d}_{tc_1})\sqcup \cdots\sqcup \supp(\bm{d}_{tc_l})\sqcup \supp(\bm{d}_{cc_j})\sqcup \cdots \sqcup\supp(\bm{d}_{cc_l}), \\ 
    &Y=\supp(\bm{y}_t)\sqcup\supp(\bm{y}_{tc_1})\sqcup \cdots\sqcup \supp(\bm{y}_{tc_l})\sqcup \supp(\bm{y}_{cc_j})\sqcup \cdots \sqcup\supp(\bm{y}_{cc_l}), \\
    &\alpha(T_{c_j})\setminus \alpha(t)\subset \sat(\bm{d}_{cc_j})\subset \supp(\bm{d}_{cc_j}) \subset \alpha(T_{c_j}), ~~~ \chi(T_{c_j})\setminus \chi(t)\subset \sat(\bm{y}_{cc_j})\subset \supp(\bm{y}_{cc_j}) \subset \chi(T_{c_j}), \\
    &\supp(\bm{d}_{tc_j})\subset \alpha(t)\cap \alpha(t_{c_j}), ~~~ 
    \supp(\bm{y}_{tc_j})\subset \chi(t)\cap \chi(t_{c_j}).
\end{align}
Using the recursive relation, we can finally compute $\Per(U^{\bm{n}}_{\bm{m}}=\bm{n}!\bm{m}!\per(\bm{n},\bm{m})$.

We now analyze the complexity.
First, at node $t$ we need to compute $\per(\bm{d}_t,\bm{y}_t)$ for all $\bm{d}_t\leq \bm{d}$ and $\bm{y}_t\leq \bm{y}$ with $|\bm{d}_t|=|\bm{y}_t|$.
Let us consider $n_1\times n_2$ matrix $U$ with a weight vector $\bm{n}$ and $\bm{m}$.
Let $\bm{d}\leq \bm{n}$ and $\bm{y}\leq \bm{m}$ with $D=\supp(\bm{d})$ and $Y=\supp(\bm{y})$ and $|\bm{d}|=|\bm{y}|=i$.
We set the upper threshold as $c$.
As previous, we use a recursive relation:
Letting $a_0$ be the first element of $D$,
\begin{align}
    \per(\bm{d},\bm{y})
    =\frac{1}{\bm{d}!\bm{y}!}\Per(U^{\bm{d}}_{\bm{y}})
    =\frac{1}{\bm{d}!\bm{y}!}\sum_{x\in Y}U_{a_0,x}\Per(U^{\bm{d}-\bm{d}_{a_0}}_{\bm{y}-\bm{y}_x})
    =\frac{1}{\bm{d}!\bm{y}!}\sum_{x\in Y}U_{a_0,x}(\bm{d}-\bm{d}_{a_0})!(\bm{y}-\bm{y}_{x})!\per({\bm{d}-\bm{d}_{a_0}},{\bm{y}-\bm{y}_x}),
\end{align}
where $\bm{d}_{a_0}$ and $\bm{y}_x$ are weight vector corresponding $a_0$ and $x$, respectively.
Thus, one can compute $\per(\bm{d},\bm{y})$ for $|\bm{d}|=|\bm{y}|=i$ using $\per(\bm{d}',\bm{y}')$ for $|\bm{d}'|=|\bm{y}'|=i-1$.
Here, the loop is over at most $n_2$ elements and finding the sets $\bm{d}-\bm{d}_{a_0}$ and $\bm{y}-\bm{y}_{x}$ costs at most $O(n_1+n_2)$.
Since the number of elements of pairs of $\bm{d}$ and $\bm{y}$ is $O((c+1)^{n_1+n_2})$, the complexity of computing permanent of all $\bm{d}$ and $\bm{y}$ is $O((n_1+n_2)^2(c+1)^{n_1+n_2})$.

Once we know all the values of $\per(\bm{d}_t,\bm{y}_t)$, $\per(\bm{d}_{tc_j},\bm{y}_{tc_j})$, and $\per(\bm{d}_{cc_j},\bm{y}_{cc_j})$, $\per(\bm{d},\bm{y})$ for all $\bm{d}$ and $\bm{y}$ can be computed in $O(lw(2c+1)^w \log c)$ by using a (non-circular) convolution \cite{elliott2013handbook}.
Therefore, the total complexity to compute $\per(\bm{n},\bm{m})$ for a $k\times k$ matrix is $O(kw(2c+1)^w\log c)$.

\subsection{Dynamical Programming Algorithm (loop hafnian)}\label{hafnian}
\subsubsection{Tree decomposition of a symmetric graph}
Now, we provide a classical algorithm using dynamical programming to compute a loop hafnian of a symmetric matrix.
We define a symmetric graph $G(B, X)$ for a symmetric matrix $B$, where vertices $X$ represent columns (or rows) of $B$ and the graph edge $ij$ if $B_{ij}\neq 0$ for $i,j\in X$.
We then define a tree decomposition of symmetric graph $G$ to compute the loop hafnian of matrix $B$, which is written as
\begin{align}
    \lHaf(B)=\sum_{\pi\in\text{PM}(G)} \prod_{ij\in\pi}B_{ij},
\end{align}
where $\text{PM}(G)$ is the set of perfect matchings of $G$.
Tree decomposition of a symmetric graph is $(T,\chi)$ such that

i) The union of $\{\chi(t)\}_{t\in T}$ is the whole column set $X$.

ii) For every edge $(x_i,x_j)$ of $G$, there exists a node $t$ of $T$ with $x_i,x_j\in\chi(t)$.

iii) For every $x\in X$, the set $\{t:x\in\chi(t)\}$ forms a subtree of $T$.

The width $w$ of a tree decomposition is the largest of $|\chi(t)|-1$ among all nodes $t$, and the treewidth of a graph is the smallest width over all possible tree decompositions of the graph.

We now define the partial loop hafnian of a subgraph $D$ of $G$,
\begin{align}
    \lhaf(D)=\sum_{\pi\in\text{PM}_G(D)}\prod_{ij\in\pi}B_{ij}.
\end{align}

\subsubsection{Collision-free cases}
Again, we first assume that $m_i\in \{0,1\}$ for all $i$'s, i.e., collision-free cases.
Note that the loop hafnian that we want to compute is written as (see Sec.~\ref{SM:GBS})
\begin{align}
    \lHaf(B_{\bm{m}})=\lhaf(X),
\end{align}
where $B_{\bm{m}}$ is obtained by repeating $i$th column and row of $B$ for $m_i$ times.

We first present the recursion relation:
\begin{lemma}
    Let $B$ be a matrix with associated symmetric graph $G$.
    Let $(T,\chi)$ be a tree decomposition of $G$.
    Let $t$ be an internal node of $T$, and let $Y$ be such that
    \begin{align}
        \chi(T_{t})\setminus \chi(t)\subset Y \subset \chi(T_{t}).
    \end{align}
    Then,
    \begin{align}
        \lhaf(Y)=\sum_{\mathcal{Y}}\lhaf(Y_{t})\prod_{j=1}^l\lhaf^*(Y_{c_{j}}),
    \end{align}
    where the sum is over all $\mathcal{Y}=(Y_{t},Y_{c_1},\dots, Y_{c_l})$ such that:
    \begin{align}\label{eq:haf1}
        Y=Y_{t}\sqcup (Y_{c_{1}}\sqcup \dots \sqcup Y_{c_{l}}),~~~
        \chi(T_{c_{j}})\setminus \chi(t)\subset Y_{c_{j}}\subset \chi(T_{c_{j}}),~~~Y_{t}\subset \chi(t).
    \end{align}
\end{lemma}
\begin{proof}
    Consider a matching $\pi\in\PM_G(Y)$ with $\chi(T_{t})\setminus \chi(t)\subset Y \subset \chi(T_{t})$.
    Consider an element of the matching, $ij\in\pi$.
    If both $i$ and $j$ are in $\chi(t)$, we let $ij\in\pi_t$ and $i,j\in Y_t$.
    Otherwise, i.e., $i\in \chi(t)$ but $j\not\in \chi(t)$, or $i\not\in \chi(t)$ but $j\in \chi(t)$, or $i,j\not\in\chi(t)$, by definition of tree decomposition, there exists a child such that $i,j\in\chi(T_c)$, and we let $ij\in \pi_c$ and $i,j\in Y_c$.
    Such a way guarantees that the decomposition satisfies the condition~\eqref{eq:haf1}.
    Then, we obtain a unique decomposition $\pi=\pi_t\sqcup \pi_{c_1}\cdots\sqcup \pi_{c_l}$.

    Conversely, for given matchings $\pi_t$ and $\pi_{c_1}, \dots, \pi_{c_l}$ defined as above, if we combine them, we can construct a matching $\pi\in\PM_G(Y)$.
    Thus,
    \begin{align}
        \lhaf(B_{\bm{m}})&=\sum_{\pi \in \PM_G(Y)}\prod_{ij\in\pi}B_{ij}
        =\sum_{\mathcal{Y}}\left(\sum_{\pi_t\in \PM_G(Y_t)}\prod_{ij\in \pi_t}B_{ij}\right)
        \left(\prod_{j=1}^l\sum_{\pi_{c_j}\in \PM_G(Y_{c_j})}\prod_{ij\in\pi_{c_j}}B_{ij}\right)\nonumber \\
        &=\sum_{\mathcal{Y}}\lhaf(Y_t)\prod_{j=1}^l\lhaf^*(Y_{c_j}),
    \end{align}
    where $\lhaf^*(Y_{c_j})$ is the sum for matchings $\pi_{c_j}$ satisfying $\pi_{c_j}\cap (\chi(t)\times \chi(t))=\emptyset$.
\end{proof}

Finally, by using inclusion-exclusion formula, similarly to permanent,
\begin{align}\label{eq:hafnian_inex}
    \sum_i\sum_{\substack{Y_{tc}\subset Y_c\cap \chi(t)\\|Y_{tc}|=i}}(-1)^i\lhaf(Y_{tc})\lhaf(Y_c\setminus Y_{tc})=\sum_{\pi_c^*}\prod_{ij\in \pi_c}B_{ij},
\end{align}
where the sum is over $\pi_c^*$'s satisfy $\pi_c^*\cap (\chi(t)\times\chi(t))=\emptyset$.
To show this, observe that for some $Y_{tc}\subset \chi(t)\cap Y_{t}$,
\begin{align}
    \lhaf(Y_{tc})\lhaf(Y_c\setminus Y_{tc})=\sum_{\pi_c(Y_{tc})}\prod_{jk\in\pi_c}B_{jk},
\end{align}
where the sum is over $\pi_c$ that contains a perfect matching for $Y_{tc}$.
Since we want to compute the sum over perfect matchings that do not contain any matching between $Y_{tc}$, using inclusion-exclusion formula, we obtain Eq.~\eqref{eq:hafnian_inex}.
Thus, using the recursive relation of Eq.~\eqref{eq:hafnian_inex}, one can compute $\lHaf(B_{\bm{m}})=\lhaf(X)$.

Let us analyze the complexity.
First, the complexity of computing $\lhaf(Y_t)$ for all $Y_t$ at node $t$ can be obtained as follows:
Let $B$ be an $n\times n$ symmetric matrix.
Let $X$ be the set of columns and consider $Y\subset X$.
Then, letting $a_0$ be the first element of $Y$,
\begin{align}
    \lhaf(Y)&=\sum_{\pi \in \PM(Y)}\prod_{ij\in\pi}B_{ij}
    =\sum_{x\in Y}B_{a_0,x}\sum_{\pi' \in \PM(Y\setminus \{a_0,x\})}\prod_{ij\in\pi'}B_{ij}
    =\sum_{x\in Y}B_{a_0,x}\lhaf(Y\setminus \{a_0,x\}).
\end{align}
For each $Y$, we loop over at most $n$ elements, and for each we need $O(n)$ to find the set $Y\setminus \{a_0,x\}$.
Thus, we can compute $\lhaf(Y)$ for all $Y$ in $O(n^2 2^{n})$.

When the values of $\lhaf(Y_t)$ and $\lhaf(Y_{c_j})$ are given, one can compute $\lhaf(Y)$ as the permanent case by using a subset convolution, the complexity of which is given by $O(lw^22^w)$.
Therefore, at node $t$, the complexity of computing all $Y$ is $O(lw^22^w)$, and the complexity of computing $\lHaf(B_{\bm{m}})$ is $O(kw^22^w)$.

\subsubsection{Collision events}
Let us consider cases with collision events.
To do that, let $B$ be a symmetric matrix and a weight vector $\bm{m}=(m_1,\dots,m_M)$ is given.
The goal is to compute the loop hafnian of $B_{\bm{m}}$, obtained by repeating $i$'s row and column for $m_i$ times.
Let $G$ be the underlying graph structure of $B$.
More explicitly, vertex set $X$ is composed of $i$'s for $m_i$ times, treating $i$'s as distinct.
We again define $\bar{X}$ to be the set of $i\in [M]$ such that $m_i>0$, i.e., removing redundancy from $X$.

We define a generalized perfect matching of $\bar{X}$ that allows repeated edges.
We represent a list of repeated edges as a vector $\tau$ indexed by $E$, and for each $ij\in E$ the entry $\tau_{ij}$ indicates the number of times that edge $ij$ appears.
The degree vector of $\tau$ is $\deg(\tau)$ with coordinates $\deg(\tau)_i\equiv \sum_{j\in[M]}\tau_{ij}$.
For a given weight vector $\bm{m}$, we define an $\bm{m}$-matching of $G$ as a list $\tau$ such that $\deg(\tau)=\bm{m}$.
Let $\PM_G(\bm{m})$ be the set of all $\bm{m}$-matchings of $G$.

Following Ref.~\cite{qi2020efficient}, we define 
\begin{align}
    B(\tau)\equiv\left(\prod_{ii\in E_l} T_{\tau_{ii}}(B_{ii})\right)\left(\prod_{ij\in E_0} (B_{ij})^{\tau_{ij}}\right),
\end{align}
where $E_l\subset E$ consists of the loops, and $E_0\equiv E\setminus E_l$, and the sequence $\{T_k(a)\}_{k\in\mathbb{N}}$ satisfies the following recursion:
\begin{align}
    T_0(a)&=1,~~~T_1(a)=a,~~~ T_k(a)=a(T_{k-1}(a)+(k-1)T_{k-2}(a)).
\end{align}
In particular, $T_k(a)$ represents the loop hafnian of a constant matrix with an element $a$.
We can now rewrite a loop hafnian as
\begin{align}\label{eq:simple_lhaf}
    \lHaf(B_{\bm{m}})=\bm{m}!\sum_{\tau \in \PM_G(\bm{m})}\frac{1}{\bm{\tau}!}B(\tau),
\end{align}
define a rescaled loop hafnian as
\begin{align}
    \lhaf(\bm{y})\equiv \frac{1}{\bm{y}!}\lhaf(B_{\bm{y}})=\sum_{\tau \in \PM_G(\bm{y})}\frac{1}{\bm{\tau}!}B(\tau).
\end{align}
To show Eq.~\eqref{eq:simple_lhaf}, we again find a way to reduce a matching of a graph $G(B, X)$ to a generalized $\bm{m}$-matching allowing multiple edges.
We can rewrite vertices of $G$ as $(i,l_i)$ by introducing $l_i\in m_i$ for redundancy.
For an edge between $(i,l_i)$ and $(j,l_j)$ of a matching $\pi$, we drop the second indices and obtain an edge of a generalized matching $\tau$.
Such a way defines a function $f:\PM_G(\bm{m})$ and for a given generalized matching $\tau \in\PM_G(\bm{m})$, one can easily check that fiber $f^{-1}(\tau)$ consists of $\bm{m}!/\bm{\tau}!$ elements.
By taking into account the loops, we finally obtain Eq.~\eqref{eq:simple_lhaf}.

Now, we present the recursive relation:
\begin{lemma}
    Let $B$ be a matrix with associated graph $G$.
    Let $(T,\chi)$ be a tree decomposition of $G$.
    Let $t$ be an internal node of $T$, and let $\bm{y}$ be such that
    \begin{align}\label{eq:haf_collision_d}
        \chi(T_t)\setminus \chi(t)\subset \sat(\bm{y})\subset \supp(\bm{y})\subset \chi(T_t).
    \end{align}
    Then,
    \begin{align}\label{eq:haf_collision_res}
        \lhaf(\bm{y})=\sum_{\bm{y}_t,\bm{y}_c}\lhaf(\bm{y}_t)\prod_{j=1}^l\lhaf^*(\bm{y}_{c_j}),
    \end{align}
    where the sum is over $(\bm{y}_t,\bm{y}_{c_1},\dots,\bm{y}_{c_l})$ such that:
    \begin{align}
        &\bm{y}=\bm{y}_t+\sum_{j=1}^l \bm{y}_{c_j}, ~~~ \supp(\bm{y}_t)\subset \chi(t), ~~~ 
        \chi(T_{c_j})\setminus \chi(t)\subset\sat(\bm{y}_{c_j})\subset \supp(\bm{y}_{c_j})\subset \chi(T_{c_j}),\label{eq:haf_col1} \\ 
        &\supp(\bm{y}_t)\sqcup \supp(\bm{y}_{c_1})\sqcup\cdots \sqcup \supp(\bm{y}_{c_l})=\supp(\bm{y}). \label{eq:haf_col2}
    \end{align}
    Here, 
    \begin{align}
        \lhaf^*(\bm{y}_{c})=\sum_{\tau_c}\prod_{ij\in\tau}B_{ij},
    \end{align}
    where the sum is over perfect matchings $\tau_c$ satisfying $\tau_c\cap (\chi(t)\times\chi(t))=\emptyset$.
\end{lemma}
\begin{proof}
    Consider $\bm{y}$ satisfying \eqref{eq:haf_collision_d} and a matching $\tau$ with $\deg(\bm{\tau})=\bm{y}$. Let $Y\equiv \supp(\bm{y})$.
    First, we define a submatching $\tau_t\equiv \tau\cap (\chi(t)\times \chi(t))$. 
    We let $i,j \in Y_t$ for such $ij$.
    On the other hand, for $ij\in\tau$, if $i\not\in \chi(t)$ and $j\in \chi(t)$, or $i\in \chi(t)$ and $j\not\in \chi(t)$, or $i\not\in \chi(t)$ and $j\not\in \chi(t)$, one can find a child $c$ such that $i,j\in \chi(T_c)$ by definition of tree decomposition.
    We then define $\tau_c$ such that $ij\in\tau_c$ and let $ij\in Y_c$.
    Therefore, we now have a unique decomposition of a matching,
    \begin{align}
        \tau=\tau_t\sqcup \tau_{c_1}\sqcup\cdots \sqcup \tau_{c_l}.
    \end{align}
    Also, the decomposition of $\bm{y}$ into $\bm{y}_t$ and $\bm{y}_{c_j}$'s following that of $Y$ guarantees the condition~\eqref{eq:haf_col1}, \eqref{eq:haf_col2}.
    Conversely, if submatchings $\tau_t$ and $\tau_{c_j}$'s are given that satisfy the above conditions, we can combine them to construct a matching $\tau$.
    Therefore, we obtain Eq.~\eqref{eq:haf_collision_res}.
\end{proof}
In a similar way as previous cases, one can rewrite
\begin{align}
    \lhaf^*(\bm{y}_c)=\sum_i \sum_{\bm{y}_{tc}:|\supp(\bm{y}_{tc})|=i}(-1)^{|\bm{y}_{tc}|}\lhaf(\bm{y}_{tc})\lhaf(\bm{y}_{cc}),
\end{align}
where the sum is over $\bm{y}_{tc}$ such that
\begin{align}
    \bm{y}_c=\bm{y}_{tc}+\bm{y}_{cc},~~~ \supp(\bm{y}_{tc})\sqcup\supp(\bm{y}_{cc})=\supp(\bm{y}_c),~~~\supp(\bm{y}_{tc})\subset \chi(t).
\end{align}
Therefore, one can rewrite \eqref{eq:haf_collision_res} as
\begin{align}
    \lhaf(\bm{y})=\sum\lhaf(\bm{y}_t)\prod_{j=1}^l\lhaf(\bm{y}_{tc_j})\lhaf(\bm{y}_{cc_j}),
\end{align}
where the sum is over $\bm{y}_t$, $\bm{y_{tc_j}}$, and $\bm{d_{cc_j}}$ such that
\begin{align}
    &\bm{y}=\bm{y}_t+\sum_{j=1}^l (\bm{y}_{tc_j}+\bm{y}_{cc_j}),~~~ \supp(\bm{y})=\supp(\bm{y}_t)\sqcup \supp(\bm{y}_{tc_1})\sqcup\cdots \sqcup\supp(\bm{y}_{tc_l})\sqcup \supp(\bm{y}_{cc_1})\sqcup\cdots \sqcup\supp(\bm{y}_{cc_l}), \\
    &\supp(\bm{y}_t)\subset\chi(t),~~~
    \supp(\bm{y}_{tc_j})\subset \chi(t)\cap \chi(t_{c_j}),~~~ \chi(T_{c_j})\setminus \chi(t)\subset \sat(\bm{y}_{cc_j})\subset \supp(\bm{y}_{cc_j}) \subset \chi(T_{c_j}).
\end{align}
Using the recursive relation, we can compute $\lHaf(B_{\bm{m}})=\lhaf(\bm{m})/\bm{m}!$.

Let us analyze its complexity.
At node $t$, all $\lhaf(\bm{y}_t)$ can be computed recursively as follows:
Let $B$ be an $n\times n$ symmetric matrix.
Consider $y\subset X$ and $\bm{y}$ such that $\supp(\bm{y})=Y$.
Then, letting $a_0$ be the first element of $Y$,
\begin{align}
    \lhaf(\bm{y})&=\frac{1}{\bm{y}!}\lHaf(B_{\bm{y}})
    =\frac{1}{\bm{y}!}\sum_{\pi\in\PM(Y)}\prod_{ij\in \pi}B_{ij}
    =\frac{1}{\bm{y}!}\sum_{x\in Y}B_{a_0,x}\sum_{\pi\in\PM(Y\setminus\{a_0,x\})}\prod_{ij\in \pi}B_{ij} \\ 
    &=\frac{1}{\bm{y}!}\sum_{x\in Y}B_{a_0,x}\lHaf(B_{\bm{y}-\bm{y}_{a_0,x}})
    =\frac{1}{\bm{y}!}\sum_{x\in Y}B_{a_0,x}(\bm{y}-\bm{y}_{a_0,x})!\lhaf(\bm{y}-\bm{y}_{a_0,x})
\end{align}
where $\bm{y}_{a_0,x}$ is the weight vector corresponding to $\{a_0,x\}$.
For each $\bm{y}$, we loop over at most $n$ elements, and for each we need $O(n)$ to find the set $Y\setminus \{a_0,x\}$.
Thus, we can compute $\lhaf(\bm{y})$ for all $\bm{y}$ in $O(n^2 (c+1)^{n})$.

On the other hand, for given values of $\lhaf(\bm{y}_t)$ and $\lhaf(\bm{y}_{c_j})$, $\lhaf(\bm{y})$ for all $\bm{y}$ can be computed using a convolution with a fast Fourier transform, the complexity of which is given by $O(lw(2c+1)^{w})$ \cite{elliott2013handbook}.
Therefore, the complexity at node $t$ is given by $\tilde{O}((2c+1)^{w})$ for all $\bm{y}$ at node $t$.


\section{Sampling algorithms}
\subsection{Clifford-Clifford algorithm (single-photon boson sampling)} \label{SM:clifford}
In this section, we recall the Clifford-Clifford algorithm to simulate single-photon boson sampling \cite{clifford2018classical}.
Let us consider single-photon state input,
\begin{align}
    |\psi_\text{in}\rangle=\prod_{j\in \mathcal{S}}\hat{a}^\dagger_j|0\rangle,
\end{align}
where $\mathcal{S}$ represents the set of the position of input single photons.
The Clifford-Clifford algorithm employs a standard Monte-Carlo method, which uses marginal probabilities.
Consider a sampling from a probability distribution of $p(\bm{r})=p(r_1,\dots,r_M)$.
The probability to obtain an outcome $(r_1,\dots,r_N)$ can be decomposed as the conditional probabilities
\begin{align}
    p(r_1,\dots,r_M)=p(r_1)p(r_2|r_1)\cdots p(r_M|r_1,\dots,r_{M-1}).
\end{align}
Thus, by sampling from $r_1$ to $r_M$ in order using conditional probabilities, which are obtained by marginal probabilities as
\begin{align}
    p(r_k|r_1,\dots,r_{k-1})=\frac{p(r_1,\dots,r_{k})}{p(r_1,\dots,r_{k-1})}.
\end{align}
we can sample from a target probability distribution as well.

In order to compute marginal probabilities of single-photon boson sampling, we introduce an expanded sample space for sampling from the photon number distribution of an output state following Ref. \cite{clifford2018classical}.
First, a photon number outcome $\bm{m}=(m_1,\dots,m_M)$, where $m_i$ represents the number of output photons in $i$th mode, can be equivalently described by $\bm{z}=(z_1,\dots,z_N)$ with $z_1\leq z_2\leq \cdots \leq z_N$, where $\bm{z}$ represents the modes where a photon is detected.
The probability to obtain $\bm{z}$ is written as
\begin{align}
    p(\bm{z})=\frac{1}{\mu(\bm{z})}|\Per (U^{\mathcal{S}}_{\bm{z}})|^2,
\end{align}
where $U^{\mathcal{S}}_{\bm{z}}$ is a matrix obtained by choosing $z_i$ rows and columns corresponding to input modes $j\in\mathcal{S}$.
Here, the number of different instances of $\bm{z}$ is $\binom{N+M-1}{N}$.
By expanding the sample space, we define
\begin{align}
    q(\bm{r})\equiv \frac{1}{N!}|\Per (U^{\mathcal{S}}_{\bm{r}})|^2,
\end{align}
where $\bm{r}=(r_1,\dots,r_N)\in[M]^N$ is an unordered tuple.
Note that $p(\bm{z})=N!q(\bm{z})/\mu(\bm{z})$ and that $\Per (U^{\mathcal{S}}_{\bm{z}})=\Per (U^{\mathcal{S}}_{\bm{r}})$.
Thus, sampling from $q(\bm{r})$ and sorting $\bm{r}$ in ascending order is equivalent to sampling directly from $p(\bm{z})$.
We expand the sampling space even further with an auxiliary array $\bm{\alpha}=(\alpha_1,\dots,\alpha_N)$, where $\bm{\alpha}$ is a permutation of $\mathcal{S}$.
By defining
\begin{align}\label{eq:SPBS_marginal}
    \phi(r_1,\dots,r_k|\bm{\alpha})=\frac{1}{k!}\left|\Per (U_{r_1,\dots,r_k}^{\alpha_1,\dots,\alpha_k})\right|^2,
\end{align}
one can prove that $q(\bm{r})=\mathbb{E}_{\bm{\alpha}}[\phi(\bm{r}|\bm{\alpha})]$ \cite{clifford2018classical},
where the expectation is taken over uniform $\bm{\alpha}$, and we use for sampling
\begin{align}
    \phi(\bm{r}|\bm{\alpha})=\phi(r_1|\bm{\alpha})\phi(r_2|r_1,\bm{\alpha})\cdots\phi(r_N|r_1,\dots,r_{N-1},\bm{\alpha}).
\end{align}
Therefore, one may perform boson sampling by computing a marginal probability $\phi(r_1\dots,r_k|\bm{\alpha})$.

\subsection{Gaussian boson sampling classical algorithm}\label{SM:GBS}
In this section, we present a Gaussian boson sampling classical algorithm and show that such a algorithm can be used to take advantage of graph structure of a given circuit.
Let us first briefly review a phase-space method to handle Gaussian states (See Refs. \cite{wang2007quantum, weedbrook2012gaussian, adesso2014continuous, serafini2017quantum} for more details about Gaussian states.).
Gaussian states are defined as ones that are described by a Gaussian distribution in phase space.
Since it follows a Gaussian distribution, an $M$-mode Gaussian state can be fully characterized by its $2M\times 2M$ (Wigner) covariance matrix $V$ and $M$-dimensional first-moment vector $d$.
The first-moment vector and covariance matrix of a given quantum state $\hat{\rho}$ are defined as $d_j=\text{Tr}[\hat{\rho}\hat{Q}_j]$ and  $V_{jk}=\text{Tr}[\hat{\rho} \{\hat{Q}_j-d_j,\hat{Q}_k-d_k \}]/2$ with a quadrature-operator vector $\hat{Q}\equiv (\hat{x}_1,\hat{p}_1,\dots, \hat{x}_M,\hat{p}_M)$, satisfying the canonical commutation relation $[\hat{Q}_j,\hat{Q}_k]=i\Omega_{jk}$, where
\begin{align}
	\Omega\equiv \mathbb{1}_M \otimes 
	\begin{pmatrix}
		0 & 1 \\ 
		-1 & 0
	\end{pmatrix}.
\end{align}
In addition, the dynamics of the quantum state under Gaussian unitary transformations can be equivalently characterized by the symplectic transformation of its covariance matrix, namely, $\hat{U}\hat{\rho}\hat{U}^\dagger\iff SVS^\text{T}$, where symplectic matrices conserve the canonical commutation relation, $S^\text{T}\Omega S=\Omega$.
A squeezing operation and a beam splitter operation are a Gaussian unitary transformation.
Therefore, we describe beam splitter arrays by using their symplectic transformation.

For Gaussian boson sampling, we first begin with a vacuum state, the covariance matrix of which is given by $\mathbb{1}_{2M}/2$.
We then apply a squeezing symplectic transformation on source modes in $\mathcal{S}$, which is written as
$S_\text{sq}=\oplus_{i=1}^M\text{diag}(e^{r_i},e^{-r_i})$
where $r_i=r$ for $i\in\mathcal{S}$ and $r_i=0$ otherwise.
We assume a momentum-squeezing operation with a real positive squeezing parameter $r>0$ without loss of generality.
Thus, the covariance matrix of the input state is written as
$V_\text{in}=S_{\text{sq}}\mathbb{1}_{2M}S_{\text{sq}}^\text{T}/2=\oplus_{i=1}^M \text{diag}(e^{2r_i},e^{-2r_i})/2$.
In addition, beam splitters are also Gaussian unitary operations, so that beam splitter operations between two modes can be characterized by symplectic matrices $S_\text{BS}$, which is formally written as
\begin{align}
    S_{\text{BS}}=
    \begin{pmatrix}
        \cos \theta & e^{i\phi}\sin \theta \\
        -e^{-i\phi}\sin \theta & \cos \theta
    \end{pmatrix} \otimes \mathbb{1}_2.
\end{align}
The symplectic matrix corresponding to given beam splitter arrays of depths $D$ can be efficiently computed by matrix multiplications of $2M\times 2M$ beam splitter symplectic matrices.

We now present more details about Gaussian boson sampling and its classical algorithm.
We consider $N$ number of sources in $M$ bosonic modes with beam splitter arrays.
We write a (complex) covariance matrix $\Sigma_{ij}=\text{Tr}[\hat{\rho}\{\hat{\xi}_i,\hat{\xi}_j\}]/2$ of the final Gaussian state $\hat{\rho}$ with $\hat{\xi}=(\hat{a}_1,\dots,\hat{a}_M,\hat{a}_1^\dagger,\dots,\hat{a}_M^\dagger)$ to follow a notational convention used in Refs. \cite{hamilton2017gaussian,quesada2020exact}.
Note that a covariance $\Sigma$ can be easily obtained by a (Wigner) covariance matrix $V$, $\Sigma=FVF^\dagger$, where $F$ changes the order of quadrature operators as $(\hat{x}_1,\hat{p}_1,\dots,\hat{x}_M,\hat{p}_M)$ to $(\hat{x}_1,\dots,\hat{x}_M,\hat{p}_1,\dots,\hat{p}_M)$ and multiply $\scriptsize{\frac{1}{\sqrt{2}}\begin{pmatrix}1&i\\1&-i \end{pmatrix}}\otimes \mathbb{1}_M$.

In the case of Gaussian states of a covariance matrix $\Sigma$ and a zero displacement, the probability of each outcome $(n_1,n_2,\dots,n_M)$ obtained by the measurement in photon-number basis is given by \cite{hamilton2017gaussian}
\begin{align}
    P(m_1,m_2,\dots,m_M)=\frac{1}{\sqrt{\text{det}(\Sigma+\mathbb{1}_{2M}/2)}}\frac{\text{Haf}(A_{\bm{m}})}{m_1!\cdots m_M!},
\end{align}
where
\begin{align}
    A=
    X_M[\mathbb{1}_{2M}-(\Sigma+\mathbb{1}_{2M}/2)^{-1}], ~~~
    X_m=
    \begin{pmatrix}
    0 & \mathbb{1}_M \\
    \mathbb{1}_M & 0
    \end{pmatrix}.
\end{align}
Here, $A_{\bm{m}}$ is a matrix obtained by repeating the $j$th and $(j+M)$th row and column of $A$ for $n_j$ times for $1\leq j \leq M$, and $\Haf(A)$ is the hafnian of a matrix $A$ \cite{barvinok2016combinatorics}.

We supply a recently proposed classical algorithm to simulate Gaussian boson sampling in Ref.~\cite{quesada2022quadratic}.
We first decompose a covariance matrix as $V_{\text{out}}=V+W$, with $V=SS^\text{T}/2$ and $W\geq0$, where $S$ is a symplectic matrix so that $V$ is a covariance matrix of a pure state.
Thus, the Gaussian state of $V_{\text{out}}$ can be interpreted as a mixture of states obtained by applying a random displacement sampled from a normal distribution of $W$ to a pure Gaussian state of $V$. 
Hence, by applying a random displacement $\mu$ sampled by a normal distribution with a covariance matrix $W$, we sample a pure Gaussian state having a covariance matrix $V$ and a mean vector $\mu$.

We then sample $(x_1,p_1,\dots,x_M,p_M)$ from a normal distribution with the covariance matrix $V+\mathbb{1}/2$ and transform the sample to a complex vector $(\alpha_1,\dots,\alpha_M)$ with $\alpha_j\equiv (x_j+ip_j)/\sqrt{2}$.
Based on the sample $(\alpha_2^*,\dots,\alpha_M^*)$, we compute conditional probability $p(m_1|\alpha_2^*,\dots,\alpha_M^*)$, which can be computed by the conditional covariance matrix and mean vector.
More explicitly, for $k\in [M]$,
\begin{align}
    P(m_1,\dots,m_k|\alpha_{k+1},\dots,\alpha_{M})=\frac{1}{\sqrt{\det(\Sigma^{(k)}+\mathbb{1}_{2k}/2)}}\frac{\lHaf(\tilde{A}_{\bm{m}}^{(k)})}{m_1!\dots m_k!}
    =\frac{1}{\sqrt{\det(\Sigma^{(k)}+\mathbb{1}_{2k}/2)}}\frac{|\lHaf(\tilde{B}_{\bm{m}}^{(k)})|^2}{m_1!\dots m_k!},
\end{align}
where $\lHaf(A)$ is the loop hafnian of a matrix $A$ \cite{bjorklund2019faster} and $\Sigma^{(k)}$ is the conditional covariance matrix for $(\alpha_{k+1},\dots,\alpha_M)$, and we have used the block structure of $A^{(k)}$ and $\tilde{A}^{(k)}$ as
\begin{align}
    A^{(k)}&=X_k[\mathbb{1}_{2k}-(\Sigma^{(k)}+\mathbb{1}_{2k}/2)^{-1}]=B^{(k)}\oplus {B^{(k)}}^*, \\ 
    \tilde{A}^{(k)}&=\text{fdiag}(A^{(k)},\gamma^{(k)})=\tilde{B}^{(k)}\oplus {\tilde{B}^{(k)^*}},
\end{align}
where $\text{fdiag}(A,\gamma)$ is a matrix filling the diagonal elements of $A$ by the vector $\gamma$, and $\gamma\equiv (\Sigma+\mathbb{1}/2)^{-1}\bm{\alpha}$ and $\gamma^{(k)}$ is obtained similarly to $A^{(k)}$.
Here, $B_{\bm{m}}^{(k)}$ is obtained by repeating $i$th column and row for $m_i$ times from marginal probabilities.
By directly computing the loop hafnian of $\tilde{B}_{\bm{m}}^{(k)}$ for $0\leq m_1\leq m_\text{max}$, we sample $n_1^*$.
Since $m_j$ can be infinitely large in principle, we choose an upper-threshold of $m_j$ carefully (See Sec. \ref{sec:threshold}).
After obtaining $m_1^*$, we discard $\alpha_2$ and continue to sample $n_2$ by computing $p(m_1^*,m_2|\alpha_3^*,\dots,\alpha_M^*)$ for $0\leq m_2\leq m^*_\text{max}$.
We continue the procedure and finally obtain a sample $(m_1^*,\dots,m_M^*)$.

We now show that such a algorithm can be exploited to take advantage of graph structure by showing that marginal probabilities also enjoy the same structure.
Let us simplify the expression $A^{(k)}$ and $B^{(k)}$.
We recall that for a block matrix of partitions $\mathcal{A}$ and $\mathcal{B}$,
\begin{align}
    M=
    \begin{pmatrix}
        M_{\mathcal{A}} & M_{\mathcal{A}\mathcal{B}} \\
        M_{\mathcal{B}\mathcal{A}} & M_{\mathcal{B}}
    \end{pmatrix},
\end{align}
the Schur complement is defined as $M/M_{\mathcal{B}}\equiv M_{\mathcal{A}}-M_{\mathcal{A}\mathcal{B}}(M_{\mathcal{B}})^{-1}M_{\mathcal{B}\mathcal{A}}$ \cite{zhang2006schur}.
Especially, the Schur complement has a property that
\begin{align}\label{eq:schur_property}
    (M^{-1})_{\mathcal{A}}=(M/M_{\mathcal{B}})^{-1}.
\end{align}
According to the algorithm above, we only need to be able to sample from a conditional (by heterodyne detection) pure Gaussian state to simulate Gaussian boson sampling.
Especially, the conditional covariance matrix from heterodyne detection on $\mathcal{B}$ is given by
$V_\mathcal{A}^{(\mathcal{B})}=(V+\mathbb{1}/2)/(V+\mathbb{1}/2)_{\mathcal{B}}-\mathbb{1}/{2}$ \cite{serafini2017quantum}.
Let us define $Q_{\mathcal{A}}^{(\mathcal{B})}\equiv \Sigma_{\mathcal{A}}^{(\mathcal{B})}+\mathbb{1}/2$.
One can easily check that we can rewrite $Q_{\mathcal{A}}^{(\mathcal{B})}=Q/Q_{\mathcal{B}}$.
When we have a Gaussian state with a covariance matrix $V$ and a mean vector $d$, the mean vector of a conditional probability after obtaining $\mu_{\mathcal{B}}$ by heterodyne detection is given by \cite{serafini2017quantum}
\begin{align}
    d^{(\mathcal{B})}_\mathcal{A}=d_\mathcal{A}+V_{\mathcal{A}\mathcal{B}}(V_{\mathcal{B}}+\mathbb{1}/2)^{-1}(\mu_{\mathcal{B}}-d_{\mathcal{B}}).
\end{align}

Now, to compute a marginal probability, we invert $Q_\mathcal{A}^{(\mathcal{B})}$ as
\begin{align}
    A_{\mathcal{A}}
    \equiv X\left[\mathbb{1}-\left(Q_\mathcal{A}^{(\mathcal{B})}\right)^{-1}\right]
    =X\left[\mathbb{1}-(Q^{-1})_\mathcal{A}\right]
    =X\left(\mathbb{1}-Q^{-1}\right)_\mathcal{A}.
\end{align}
Here, we have used the Schur complement's property of Eq.~\eqref{eq:schur_property}, $\left(Q_\mathcal{A}^{(\mathcal{B})}\right)^{-1}=(Q/Q_{\mathcal{B}})^{-1}=(Q^{-1})_{\mathcal{A}}$.
After rewriting $Q=(U\oplus U^*)T(U\oplus U^*)^\dagger+\mathbb{1}/2$, we can find $Q^{-1}=(U\oplus U^*)(T+\mathbb{1}/2)^{-1}(U\oplus U^*)^\dagger$,
where 
\begin{align}
T=
\begin{pmatrix}
    \oplus_{i=1}^M \sinh^2r_i & \oplus_{i=1}^M \sinh{r_i}\cosh{r_i} \\ 
    \oplus_{i=1}^M \sinh{r_i}\cosh{r_i} & \oplus_{i=1}^M \sinh^2r_i
\end{pmatrix},~~~
    (T+\mathbb{1}/2)^{-1}=
    \begin{pmatrix}
        \mathbb{1} & -\oplus_{i=1}^M \tanh r_i \\ 
        -\oplus_{i=1}^M \tanh r_i & \mathbb{1}
    \end{pmatrix}.
\end{align}
Thus, $A_{\mathcal{A}}=B_{\mathcal{A}}\oplus B_{\mathcal{A}}^*$,
with $B_{\mathcal{A}}=[U(\oplus_{i}\tanh{r_i})U^\text{T}]_{\mathcal{A}}$.
Finally, $\tilde{A}_\mathcal{A}=\text{fdiag}(A_\mathcal{A},\gamma_\mathcal{A})$ and $\tilde{B}_\mathcal{A}=\text{fdiag}(B_\mathcal{A},\gamma_\mathcal{A})$.
Here, $\gamma_\mathcal{A}=[Q^{-1}]_\mathcal{A}\bm{\alpha}_\mathcal{A}^{(\mathcal{B})}$ and $\bm{\alpha}_{\mathcal{A}}^{(\mathcal{B})}$ is a complex mean vector transformed from $d_{\mathcal{A}}^{(\mathcal{B})}$.
Since marginal probabilities depend on the loop hafnian of $\tilde{B}_{\mathcal{A}}$ and its graph can be obtained by simply discarding vertices for part $\mathcal{B}$ (conditioning part), marginal probabilities follow the graph structure of a probability.

\subsection{Proof of Theorem 1}\label{SM:proofth1}
To simulate single-photon boson sampling, we implement the Clifford-Clifford algorithm by computing marginal probabilities $\phi(r_1,\dots,r_k|\bm{\alpha})$ with $k=1,\dots,N$. 
At each $k\in [N]$, we compute $M$ marginal probabilities, corresponding to $\Per(U^{\alpha_1,\dots,\alpha_k}_{r_1,\dots,r_k})$, so the maximal complexity is $O(MNw^22^w)$.
One can easily see that the bipartite graphs corresponding to marginal probabilities are minors of a bipartite graph for a probability.
Since we iterate for $k\in [N]$, the total complexity of sampling is at most $O(MN^2w^22^w)$. 
For Gaussian boson sampling, at each step $k\in [M]$, we need to compute total two marginal probabilities corresponding to the conditional covariance matrices, i.e., zero and a single photon, and each of those is a loop hafnian of $k \times k$ matrix, having the maximal complexity $O(Nw^22^w)$.
Thus, the total complexity is $O(MNw^22^w)$.

\section{Error of photon-number truncation}\label{sec:threshold}
In this section, we analyze the influence of photon-number truncation in squeezed states.
When we have $N$ squeezed vacuum states and measure them in photon-number basis, 
the probability to generate a total of $k$ photon pair events (2$k$ photons) is given by the negative binomial distribution \cite{hamilton2017gaussian, kruse2019detailed}
\begin{align}
    P_N(k)=\binom{\frac{N}{2}+k-1}{k}\text{sech}^Nr\tanh^{2k}r.
\end{align}
Note that beam splitter arrays do not change the probability distribution of total photon numbers.
The tail probability of the negative binomial distribution is given by \cite{brown2011wasted}
\begin{align}
    \text{Pr}(k> \alpha N \text{sech}^2 r)\leq \exp\left[\frac{-\alpha N(1-1/\alpha)^2}{2}\right].
\end{align}
While increasing a constant $\alpha$ decreases the truncation error exponentially, the order of the complexity of sampling does not change.
Any accuracy can be achieved by increasing $\alpha$ with a constant factor which reduces error exponentially.
In order to make the truncation error to be smaller than $\epsilon$ for sufficiently large $\alpha$, we may choose $\alpha N \text{sech}^2r= 2\text{sech}^2r\log(1/\epsilon)\equiv m_\text{max}$.

\section{Approximation method}\label{sec:approx}
\subsection{Approximation of a unitary matrix}\label{SM:W}
Let $U$ be a true unitary transformation matrix.
Our approximation strategy is $U\to \tilde{U}\equiv U-dU$, where we have removed unitary matrix elements for jump more than diffusive dynamics.
Note that $\tilde{U}$ is no longer unitary in general.
We provide an approximation algorithm and the upper bound of its simulation error.

We first extend $\tilde{U}$ to a unitary matrix in $2M\times 2M$.
Using singular value decomposition, $\tilde{U}=RDV$,
we transform $\tilde{U}\to \bar{U}=\tilde{U}/\kappa$ to make the singular values smaller than $1$, i.e., $\kappa$ is chosen to be the maximum singular value or 1 if the maximum singular value is smaller than 1, and define $\bar{U}\equiv R\bar{D}V=R\tilde{D}V/\kappa$.
Note that in practice one may merely replace singular values larger than 1 by 1, while we divide them by the maximum for simplicity of the proof.
By Mirsky's theorem \cite{mirsky1960symmetric}, we have
\begin{align}
    (\tilde{\sigma}_1-\sigma_1)^2\leq\sum_i(\tilde{\sigma}_i-\sigma_i)^2\leq \|dU\|_F^2
\end{align}
where $\sigma_i$'s are singular values of $U$, i.e., $\sigma_i=1$ and $\tilde{\sigma}_i$'s are singular values of $\tilde{U}$ in descending order.
Thus, $(\kappa-1)^2\leq \|dU\|_F^2$.
Defining $\mu=\kappa-1$, $\mu^2\leq \|dU\|_F^2$.
Now we extend the matrix $\bar{U}$ to a $2M\times 2M$ unitary matrix
\begin{align}
    W=
    \begin{pmatrix}
        \bar{U} & R\sqrt{1-\bar{D}^2}V \\
        R\sqrt{1-\bar{D}^2}V & -\bar{U}
    \end{pmatrix}.
\end{align}
We also extend a true unitary matrix into a $2M\times 2M$ matrix as $U_{2M}\equiv U\oplus (-U)$.
One can check that $U_{2M}$ and $W$ are close if the approximation of $U$ is small:
\begin{align}
    \|dW\|^2_F \equiv \|W-U_{2M}\|_F^2=\text{Tr}(W-U_{2M})^\dagger(W-U_{2M})=2\|\bar{U}-U\|_F^2+\text{Tr}(1-\bar{D}^2),
\end{align}
where
\begin{align}
    \|\bar{U}-U\|_F^2
    =\left\|\frac{\tilde{U}}{\mu+1}-U\right\|_F^2\leq\|U-dU-(1+\mu) U\|_F^2
    =\|dU+\mu U\|_F^2
    \leq (\sqrt{M}+1)^2\|dU\|_F^2,
\end{align}
and
\begin{align}
    \text{Tr}(1-\bar{D}^2)&=\text{Tr}\left(1-\frac{\tilde{D}^2}{(1+\mu)^2}\right)
    =\sum_{i=1}^M\left(1+\frac{\tilde{\sigma}_i}{1+\mu}\right)\left(1-\frac{\tilde{\sigma}_i}{1+\mu}\right)
    \leq \sqrt{\sum_{i=1}^M\left(1+\frac{\tilde{\sigma}_i}{1+\mu}\right)^2\sum_{i=1}^M\left(1-\frac{\tilde{\sigma}_i}{1+\mu}\right)^2} \\
    &\leq \sqrt{\sum_{i=1}^M\left(1+\frac{\tilde{\sigma}_i}{1+\mu}\right)^2}(\sqrt{M}+1)\|dU\|_F
    \leq 2\sqrt{M}(\sqrt{M}+1)\|dU\|_F.
\end{align}
Here, the first inequality follows from the Cauchy-Schwarz inequality, and the second inequality follows from the Mirsky's theorem and the fact that $\tilde{\sigma_i}/(1+\mu)\leq 1$ is the singular values of $\bar{D}$.
Hence,
\begin{align}\label{eq:dwdu}
    \|dW\|^2_F\leq 2(\sqrt{M}+1)^2\|dU\|_F^2+2\sqrt{M}(\sqrt{M}+1)\|dU\|_F\leq 2(\sqrt{M}+1)^2(\|dU\|_F^2+\|dU\|_F).
\end{align}

\subsection{Single-photon state approximation}
Our approximation algorithm operates as follows:
We first approximate a unitary matrix by a nonunitary matrix $\tilde{U}=U-dU$ and obtain the extended $2M\times 2M$ unitary matrix $W$ as Sec. \ref{SM:W}.
We then employ the Clifford-Clifford algorithm.
Running the algorithm, at $k$th step, we compute $\phi(r_1,\dots,r_k|\bm{\alpha})$ as Eq.~\eqref{eq:SPBS_marginal} for $r_k$ from 1 to $M$ for given $r_1,\dots,r_{k-1}$.
For those cases, the graph structure of an approximated matrix $\tilde{U}$ can be employed so that it can be computed by using dynamical programming.
More specifically, the relevant bipartite graph has vertices for the sources and $(r_1,\dots,r_k)$, depends only on $\tilde{U}$.
Since we do not cover all outcomes $r_k$ from $M+1$ to $2M$, the probability sum over $r_k=1,\dots,M$ can be less than zero, meaning that we have a chance to detect photons outside of the first $M$ modes.
For these cases, we return ``out", and otherwise we continue.
Such a way guarantees that we obtain photons only at the first $M$ modes as we desire.

Let us analyze the total variation distance:
\begin{align}\label{eq:sap}
    \text{TVD}
    &= \frac{1}{2}\sum_{x=x_M\oplus \mathbb{0}}|P_{\text{ideal}}(x)-P_{\text{approx}}(x)|+\frac{1}{2}\left|\sum_{x\neq x_M\oplus \mathbb{0}}\left(P_{\text{ideal}}(x)-P_{\text{approx}}(x)\right)\right| \\ 
    &
    \leq \frac{1}{2}\sum_{x}|P_{\text{ideal}}(x)-P_{\text{approx}}(x)|
    \leq \frac{N}{2}\|dW\|_F.
\end{align}
where $x_M\oplus \mathbb{0}$ represents outcomes for which photon numbers are zero other than the first $M$ modes.
The second inequality is from Ref. \cite{arkhipov2015bosonsampling}, which states that the TVD is bounded by the operator norm of two unitary transformation matrices, and the last inequality follows from the definition of the operator norm and Frobenious norm.

\subsection{Gaussian state approximation}
For Gaussian boson sampling as well, we first extend an approximated matrix $\tilde{U}=U-dU$ to $W$ and transform it to a symplectic matrix.
Our Gaussian state's covariance matrix is then written as
\begin{align}
    \bar{V}_{2M}=S_W(D\oplus \mathbb{1}_M/2)S_W^\text{T},
\end{align}
where the added $\mathbb{1}_M/2$ represent vacuum.
The true extended covariance matrix can be written as
\begin{align}
    V_{2M}=(S_U\oplus S_{-U})(D\oplus \mathbb{1}_M/2)(S_U\oplus S_{-U})^\text{T}.
\end{align}
First, we show that the distance between the approximated and true covariance matrices are bounded by the Frobenius norm of unitary matrices,
\begin{align}
    \|V_{2M}-\bar{V}_{2M}\|_F&=\|SDS^T-\bar{S}D\bar{S}^\text{T}\|_F\leq 2\|S-\bar{S}\|_F\|D\|_F\|S\|_F
    =\|S-\bar{S}\|_F\sqrt{2M[N\cosh{4r}+(M-N)]}\nonumber \\ 
    &\leq2\|dW\|_F\sqrt{M[N\cosh{4r}+(M-N)]}.
\end{align}

Finally, since we only care about the output photons at the first $M$ modes, we postselect such outcomes.
This can be implemented by first computing the probability to obtain nonzero photons for the additional $M$ modes and sample from the binary.
If photons click for the modes, then we return ``out", otherwise we proceed the algorithm by using conditional covariance matrix for the first $M$ modes.
Such a way leads to the TVD as
\begin{align}\label{eq:gap}
    \text{TVD}
    &=\frac{1}{2}\sum_{x_M\oplus \mathbb{0}}|P_{\text{ideal}}(x)-P_{\text{approx}}(x)|
    \leq \frac{1}{2}\sum_{x_{2M}}|P_{\text{ideal}}(x)-P_{\text{approx}}(x)|
    \leq \sqrt{1-F(V_{2M},\bar{V}_{2M})} \nonumber \\
    &\leq (N\cosh{4r}/2)^{1/4}\|V_{2M}-\bar{V}_{2M}\|_F^{1/2}\leq \|dW\|_F^{1/2} \poly(N).
\end{align}
Here, the second and third inequalities are found as follows:
Recall that total variation distance can be bounded by quantum infidelity $1-F$ \cite{fuchs1999cryptographic},
\begin{align}
\frac{1}{2}\sum_{x} |P(x)-P_a(x)|\leq \frac{1}{2}\| \hat{\rho}-\hat{\rho}_a \|\leq \sqrt{1-F(\hat{\rho},\hat{\rho}_a)}.
\end{align}
Quantum fidelity between two $M$-mode Gaussian states characterized by covariance matrices $V_1,V_2$ and zero means, one of which is pure, can be written as \cite{spedalieri2012limit, banchi2015quantum}
\begin{align}
    F(V_1,V_2)=\frac{1}{\sqrt{\det(V_1+V_2)}},
\end{align}
where we used a covariance matrix instead of quantum state $\hat{\rho}$ in the argument because a covariance matrix completely characterizes a quantum state in our case.
The following lemma shows that the quantum infidelity between two Gaussian states characterized by $V_1$ and $V_2$ can be bounded by the Frobenius norm $\|\cdot\|_F$ of their difference matrix $X=V_1-V_2$.
\begin{lemma}\label{lemma:infid}
    Let $V_1$ be a covariance matrix of a Gaussian state in bosonic modes $\mathcal{M}$ obtained by applying beam splitter arrays on single-mode squeezed states of squeezing parameter $r$ and $V_2$ be a covariance matrix of a Gaussian state.
    For small $\|X\|_F$, the quantum infidelity between the two Gaussian states is bounded by $1-F(V_1,V_2)\leq \|X\|_F\sqrt{N\cosh{4r}/2}$.
\end{lemma}
\begin{proof}
Note that the covariance matrix $V_1$ can be decomposed as $V_1=S(\mathbb{1}_{2M}/2)S^\text{T}$ by a symplectic matrix $S$ satisfying $S\Omega S^\text{T}=\Omega$ and that the symplectic matrix can be decomposed as $S=O S_\text{sq}$, where $O$ represents a symplectic matrix corresponding to beam splitter arrays and $S_\text{sq}$ represents squeezing operators to generate squeezed vacuum sources.
We then have
\begin{align}
\det(V_1+V_2)&=\det(2V_1-X)=\det(\mathbb{1}_{2M}-S^{-1} X S^{-\text{T}})
\leq \left(1+\frac{1}{2M}|\text{Tr}[S^{-1} X S^{-\text{T}}]|\right)^{2M} \nonumber \\
&\leq \left(1+\frac{1}{2M}\|X\|_F\|S^{-\text{T}}S^{-1}\|_F\right)^{2M}\leq \left(1-\frac{1}{2}\|X\|_F\|S^{-\text{T}}S^{-1}\|_F\right)^{-2},
\end{align}
where for the first inequality, we have used the symmetric property of $X$ and the inequality of arithmetic geometric means and the fact that
\begin{align}
\text{Tr}[S^{-1}X S^{-\text{T}}]&=\text{Tr}[S^{-1}(V_1-V_2) S^{-\text{T}}]
=\text{Tr}[\mathbb{1}_{2M}/2-S^{-1}V_2S^{-\text{T}}]
=\text{Tr}[\mathbb{1}_{2M}/2-\tilde{V}_2]\leq 0.
\end{align}
The Cauchy-Schwarz inequality has been used for the second inequality.
For the last inequality, we used $(1+x/M)^M\leq (1-x)^{-1}$ for $0\leq x \leq 1$.
Then, the quantum fidelity between Gaussian states with covariance matrices $V_1$ and $V_2$ is approximated as
\begin{align}\label{eq:fid_bound}
1-F(V_1,V_2)&=1-\frac{1}{\sqrt{\det(V_1+V_2)}}
\leq \frac{1}{2}\|X\|_F\|S^{-T}S^{-1}\|_F
=\sqrt{N\cosh{4r}/2}\|X\|_F.
\end{align}
\end{proof}

\subsection{Proof of Theorem 2}
According to Theorem 1, we can perform exact samplings in $O(MN^2w^22^w)$ (single-photon boson sampling up to threshold) and in $O(MNw^22^w)$ (Gaussian boson sampling) for collision-free cases. Therefore, by Eq.~(\ref{eq:sap}) and Eq.~(\ref{eq:gap}), we can do approximate samplings in the same complexity with error $O(\poly(N)\|dW\|_F$) in terms of the total variation distance. Finally, by Eq.~(\ref{eq:dwdu}), $\|dW\|_F \leq \sqrt{\|dU\|^2_F+\|dU\|_F}$ up to $\poly(N)$, thus this proves the theorem.

\section{Diffusive dynamics of random beam-splitter arrays}
\subsection{Classical random walk behavior of random beam-splitter arrays}\label{SM:randomwalk}
Let us consider two modes $\hat{a}_k^{(D)\dagger}$ and $\hat{a}_k^{(D+1)\dagger}$ at depth $D$, which are written as
\begin{align}
    \hat{a}_k^{(D)\dagger}=\sum_{j=1}^M U_{k,j}^{(D)}\hat{a}_j^{(0)},~~~ \text{and}~~~ \hat{a}_{k+1}^{(D)\dagger}=\sum_{j=1}^M U_{k+1,j}^{(D)}\hat{a}_j^{(0)\dagger}.
\end{align}
After applying a beam splitter between them, the modes are transformed as
\begin{align}
    \hat{a}_k^{(D+1)\dagger}&=e^{i\phi_1}\left(\cos\theta\sum_{j=1}^M U_{k,j}^{(D)}\hat{a}_j^{(0)\dagger}+e^{i\phi_0}\sin\theta\sum_{j=1}^M U_{k,j}^{(D)}\hat{a}_j^{(0)\dagger} \right), \\
    \hat{a}_{k+1}^{(D+1)\dagger}&=e^{i\phi_2}\left(\cos\theta\sum_{j=1}^M U_{k+1,j}^{(D)}\hat{a}_j^{(0)\dagger}-e^{-i\phi_0}\sin\theta\sum_{j=1}^M U_{k+1,j}^{(D)}\hat{a}_j^{(0)\dagger} \right),
\end{align}
where $\cos\theta$ and $\sin\theta$ represent the beam splitter's transmissivity and reflectivity, respectively, and we obtain the following relation
\begin{align}
    U_{k,j}^{(D+1)}&=e^{i\phi_1}\left(U_{k,s}^{(D)}\cos\theta+e^{i\phi_0}U_{k+1,s}^{(D)}\sin\theta\right), \\ 
    U_{k+1,j}^{(D+1)}&=e^{i\phi_2}\left(U_{k,s}^{(D)}\cos\theta-e^{-i\phi_0}U_{k+1,s}^{(D)}\sin\theta\right).
\end{align}
After averaging the transmissivity $\cos\theta$ over the uniform distribution of $\theta\in[0,2\pi)$, and the phases $\phi_0,\phi_1$ and, $\phi_2$ over the uniform distribution of $[0,2\pi)$, we obtain
\begin{align}
    \mathbb{E}[|U_{k,s}^{(D+1)}|^2]=
    \mathbb{E}[|U_{k+1,s}^{(D+1)}|^2]
    =\frac{\mathbb{E}[|U_{k,s}^{(D)}|^2]+\mathbb{E}[|U_{k+1,s}^{(D)}|^2]}{2},
\end{align}
which shows that the transmission and reflection rate of a random beam-splitter array follows a random walk behavior.

\subsection{Proof of Lemma 2}\label{SM:leakage}
We provide the proof of Lemma 2 in this section.
We first note that on average the random beam-splitter circuits can be characterized by a symmetric random walk and that the goal is to find an upper-bound on the leakage rate.
We observe that the leakage rate assuming an infinite number of modes is always larger than one with boundaries because boundaries makes the walker return to the initial lattice.
Thus, it is sufficient to find an upper-bound assuming an infinite number of modes.

For one-dimensional random walk, the probability of propagating farther than $l$ in step $t$ is given by
\begin{align}
    P\leq 2\exp\left(-\frac{l^2}{2t}\right).
\end{align}
Thus, the leakage rate for $d$-dimensional case can be upper-bounded as
\begin{align}
    \mathbb{E}[\eta]&\leq 1-\left[1-2 \exp\left(-\frac{l^2}{2t}\right)\right]^d
    \leq 2d \exp\left(-\frac{l^2}{2t}\right),
\end{align}
where $\eta=\eta(\vec{\theta}, \vec{\phi})$ is a function of parameters $(\vec{\theta}, \vec{\phi})$ for beam splitter arrays, which are random variables following a uniform distribution on $[0,2\pi)$.

Using Markov's inequality, we obtain
\begin{align}
    P(\eta\geq a)\leq \frac{\mathbb{E}[\eta]}{a}=\frac{2d}{a} \exp\left(-\frac{l^2}{2t}\right).
\end{align}

Since we are interested in the leakage rate of a source, we set $l=\kappa L=\kappa k^{1/d}N^{(\gamma-1)/d}$ with $L=(M/N)^{1/d}$ for a $d$-dimensional case.
Taking into account the dimension of the circuit, we set $t=D/d$.
Especially for depth $D\leq dk^{2/d}c_1 \kappa^2N^{\frac{2(\gamma-1)}{d}-\epsilon}/2$ and $a=\exp(-N^{\epsilon})$ with an arbitrary $\epsilon>0$, we find
\begin{align}
    P[\eta\geq \exp(-N^{\epsilon})]\leq 2d \exp\left[\left(1-\frac{1}{c_1}\right)N^{\epsilon}\right].
\end{align}
Thus, the probability is exponentially small in $N$ when $c_1<1$.
In other words, except for an exponentially small probability for $(\vec{\theta},\vec{\phi})$, the leakage rate $\eta(\vec{\theta},\vec{\phi})$ is exponentially small.

In addition, for $d=1$, if we set $l=\kappa L \log N =\kappa k N^{\gamma-1}\log N$, we can similarly obtain for $D\leq c_1 k^2\kappa^2 N^{2(\gamma-1)-\epsilon}(\log N)^2/2$
\begin{align}
    P[\eta\geq \exp(-N^{\epsilon})]\leq 2d \exp\left[\left(1-\frac{1}{c_1}\right)N^{\epsilon}\right].
\end{align}

\section{Proof of Theorem 3}
In this section, we provide the proof of Theorem 3. First, let us consider single-photon boson sampling. By the assumption, $D\leq D_\text{easy}\equiv dk^{2/d}N^{2(\gamma-1)/d-\epsilon}/8$. 
It corresponds $\kappa=1/2$, so that $l=L/2$; thus, the sources are confined at the initial sublattices under the approximate matrix $\tilde{U}$. 
Since the photons are detected at their initial sublattices for typical circuits, the corresponding outcomes are represented by disjoint graphs with two vertices (input photon and output photon for each sublattice). 
Therefore, the graph has treewidth $w=1$. 
Thus the complexity is $O(MN^2)$ by Theorem 1. Now we check whether the error satisfies the condition $O(1/\poly(N))$. We start with the total variation distance as
\begin{align}
    \text{TVD}\leq \poly(N)\sqrt{\|dU\|^2_F+\|dU\|_F} \leq O(1/\poly(N)),
\end{align}
where we use Theorem 2 for the first equality, and Lemma 2 for the second inequality. 
Therefore, the approximate boson sampling can be efficiently performed. 

Now, let us move to the Gaussian boson sampling. The only difference is that a Gaussian source, i.e., squeezed vacuum state, can emit multiple photons so that output can be many photons even confined in the same sublattice. We still assume that this multiplicity is at most a constant $c$ without loss of generality. Then in a single sublattice, at most $c$ photons can be detected, which leads the treewidth $c-1$ because it yields the complete graph in the symmetric tree decomposition. Although each sublattice has a complete graph, there are only disjoint graphs, and total treewidth is bounded, i.e., $w \leq c-1$. Consequently, the sampling complexity is $O(MN)$ by Theorem 1, and the error analysis is the similar to the case of single-photon boson sampling.

\begin{figure}[t]
\includegraphics[width=300px]{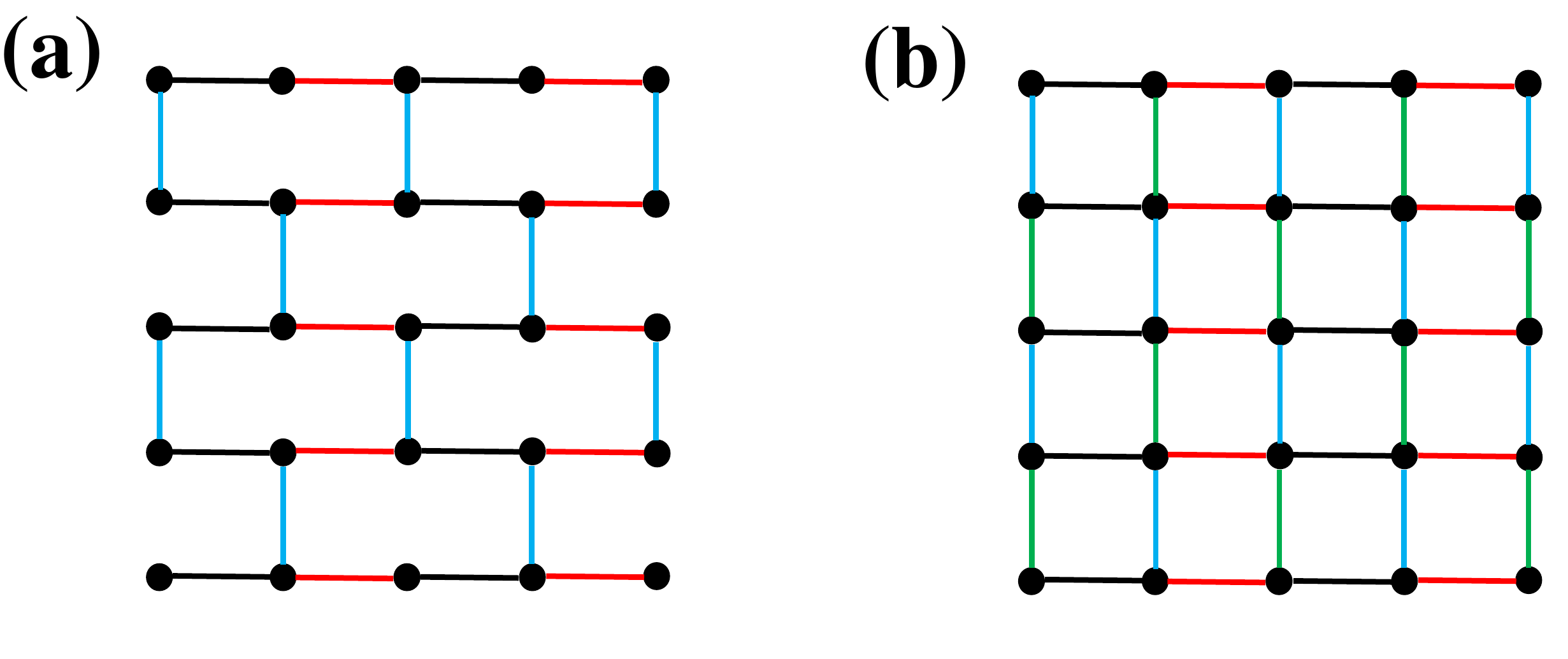}
\caption{Examples of 2D constant-depth circuit diagrams. First depth (black), second (red), third (blue), and fourth one (green). Examples of $5 \times 5$ circuit appearing (a) elementary wall in a depth-3 circuit, (b) grid graph in a depth-4 circuit. Both have unbounded treewidth $\Theta(N)$.}
\label{fig:grid}
\end{figure}

\section{Constant-depth 2D circuit}\label{SM:const}
In this section, we briefly show that for 2D structures, only constant-depth is enough to achieve unbounded treewidth of induced graphs. 
Let us focus on a dense circuit, i.e., $L=1$, in a 2D lattice, with local interactions. 
We can draw 2D circuit diagrams Fig.~\ref{fig:grid}, in which vertices represent input positions of the circuit, and edges correspond the connectivity by local interaction. 
In Fig.~\ref{fig:grid} (a), a graph with unbounded treewidth already appears for depth-3, called elementary wall~\cite{reed2003recent}. 
If we step one depth more, we can obtain a grid (Fig.~\ref{fig:grid} (b)). 
Both have the treewidth $w=\Theta(N)$ with $N$ is the height of the wall or length of the grid. 
One can find that there exist output configurations including the graph of circuit diagrams, which indicates that the treewidths of circuit diagrams are lower bounds on the maximal treewidths of the induced graphs over possible outcomes. Thus, the corresponding bipartite graph (symmetric graph) for single-photon (Gaussian) boson sampling has at least this unbounded treewidth.

It is worth mentioning that the simulability of constant-depth circuits in the literature. For  linear optical circuits with single-photon inputs, all depth-2 circuits are easy, some depth-4 circuits are hard, and the depth-3 case is unknown~\cite{constantdepthpra}. By our argument, some depth-3 circuits are hard using our treewidth-based algorithm. In addition, for a 2D Gaussian local random boson sampling, constant-depth hardness is suggested under certain conjectures~\cite{deshpande2021quantum}. For  qubit circuits, there exist depth-3 circuits with treewidth $\Omega(n)$ by using expander graph~\cite{markov2008simulating}. 

\begin{figure}[t]
\includegraphics[width=400px]{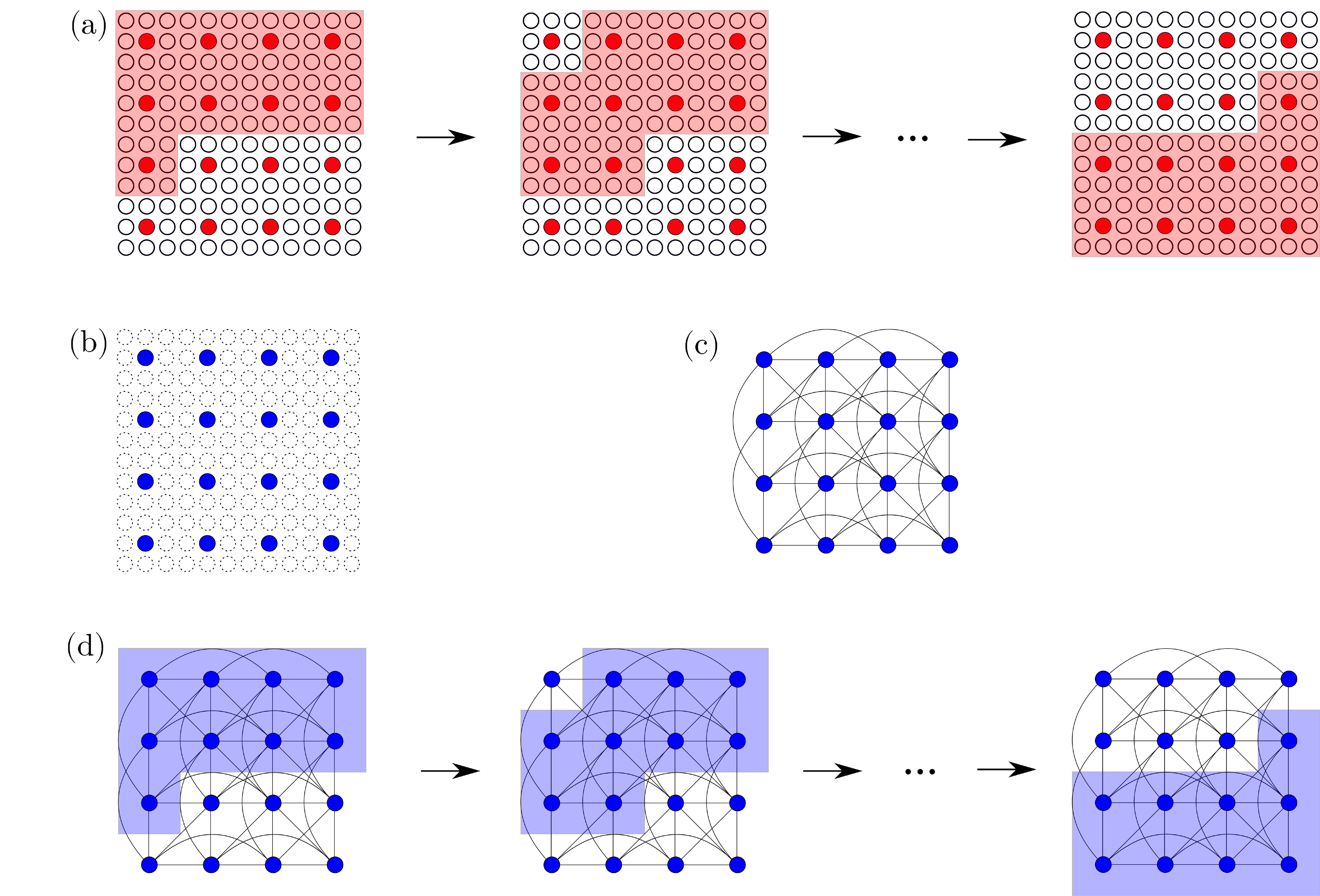}
\caption{(a) Virtual tree decomposition including all modes, which will be used to construct real bags of a symmetric tree decomposition. The former is the latter's child. Here, we assume that a photon from source can jump up to 3 modes away. $M=144$ and $N=16$ (b) An example output distribution. (c) Corresponding symmetric graph. (d) True tree decomposition taking into account true outcomes only. An upper bound on the treewidth is $w=8$.}
\label{fig:band_sym}
\end{figure}

\begin{figure}[t]
\includegraphics[width=400px]{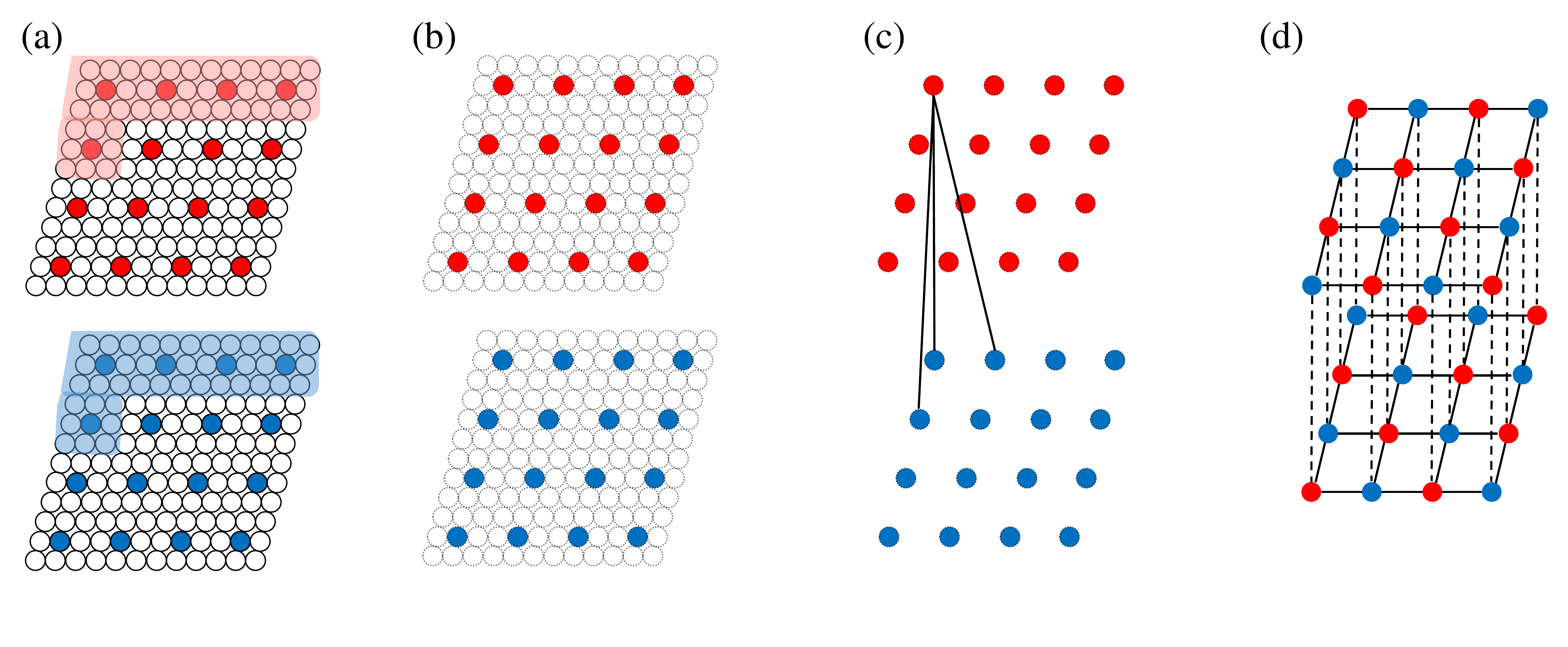}
\caption{(a) Virtual tree decomposition including all modes (red: input modes, blue: output modes), which will be used to construct real bags of a bipartite tree decomposition. Here, we assume that a photon from source can jump up to 3 modes away. The colored region represents the first bag. $M=144$ and $N=16$ (b) An example output distribution. (c) Corresponding bipartite graph. We depict edges only from the first input mode for simplicity. (d) An equivalent graph including grids by rearrangement of the vertices. An upper bound on the treewidth is $w=9$.}
\label{fig:band_bi}
\end{figure}

\section{Upper bound on the treewidth of induced graphs (Proof of Theorem 4)} 
In this section, we explicitly give examples of boson sampling in 2D and prove Theorem 4 by finding the upper bound of treewidth we need. 
Fig.~\ref{fig:band_sym} shows an example of a 2D Gaussian boson sampling and its tree decomposition.
Fig.~\ref{fig:band_sym} (a) shows a possible virtual tree decomposition before we obtain an output photon distribution.
When output photons are clicked as shown in Fig.~\ref{fig:band_sym} (b), the corresponding symmetric graph is given by Fig.~\ref{fig:band_sym} (c).
Finally, we can find the tree decomposition of an output distribution by obtaining an overlap of the virtual tree decomposition and output photons.
Here, depending on how far a photon can propagate, we have to adjust the size of each bag accordingly.
When a photon can propagate up to $\kappa L$ (assuming $\kappa$ is a nonzero integer for simplicity), for the virtual tree decomposition, the size of a bag is upper-bounded by $\sqrt{M}\times (2\kappa+1) L$, and it contains at most $\sqrt{M}/L \times (2\kappa+1) L/L=(2\kappa+1) N^{1/2}$ sources.
One can easily check that even if photons outside of a bag propagate into the bag, the number of photons that can be clicked has the same scaling as the number of sources.
Thus, the width is given by $w=\Theta(\kappa N^{1/2})$.
If we let $\kappa=\Theta(N^{\alpha/2})$, $w=\Theta(N^{\frac{1+\alpha}{2}})$.
The latter corresponds to the diffusive dynamics with $D=\Theta(N^{\alpha}D_\text{easy})$. Fig.~\ref{fig:band_bi} is for a similar case but single-photon boson sampling and bipartite tree decomposition. We remark here that our description represents an upper bound on the treewidth. The exact treewidth is hard to find in general, but we illustrate that for the 2D grid case, the band decomposition is enough for the exact treewidth, as shown in Fig.~\ref{fig:grid-tree}.

Finally, let us consider an arbitrary dimension $d$.
One may find a similar virtual tree decomposition.
For example, if a photon can jump up to $\kappa L$, the size of bags is upper-bounded by $M^{1/d}\times M^{1/d}\times \cdots \times (2\kappa+1)L=M^{(d-1)/d}\times (2\kappa+1)L$, and the number of sources is at most $(M^{\frac{d-1}{d}}\times (2\kappa+1)L)/L^d=\Theta(\kappa N^{\frac{d-1}{d}})$.
Therefore, the width is $\Theta(\kappa N^{\frac{d-1}{d}})$.
When $\kappa=\Theta(N^{\alpha/d})$, ($D=\Theta(N^{2\alpha/d} D_\text{easy})$  for diffusive dyanmics), the width is given by $\Theta(N^{\frac{\alpha}{d}+\frac{d-1}{d}})$.
It proves Theorem 4.

\begin{figure}[t]
\includegraphics[width=300px]{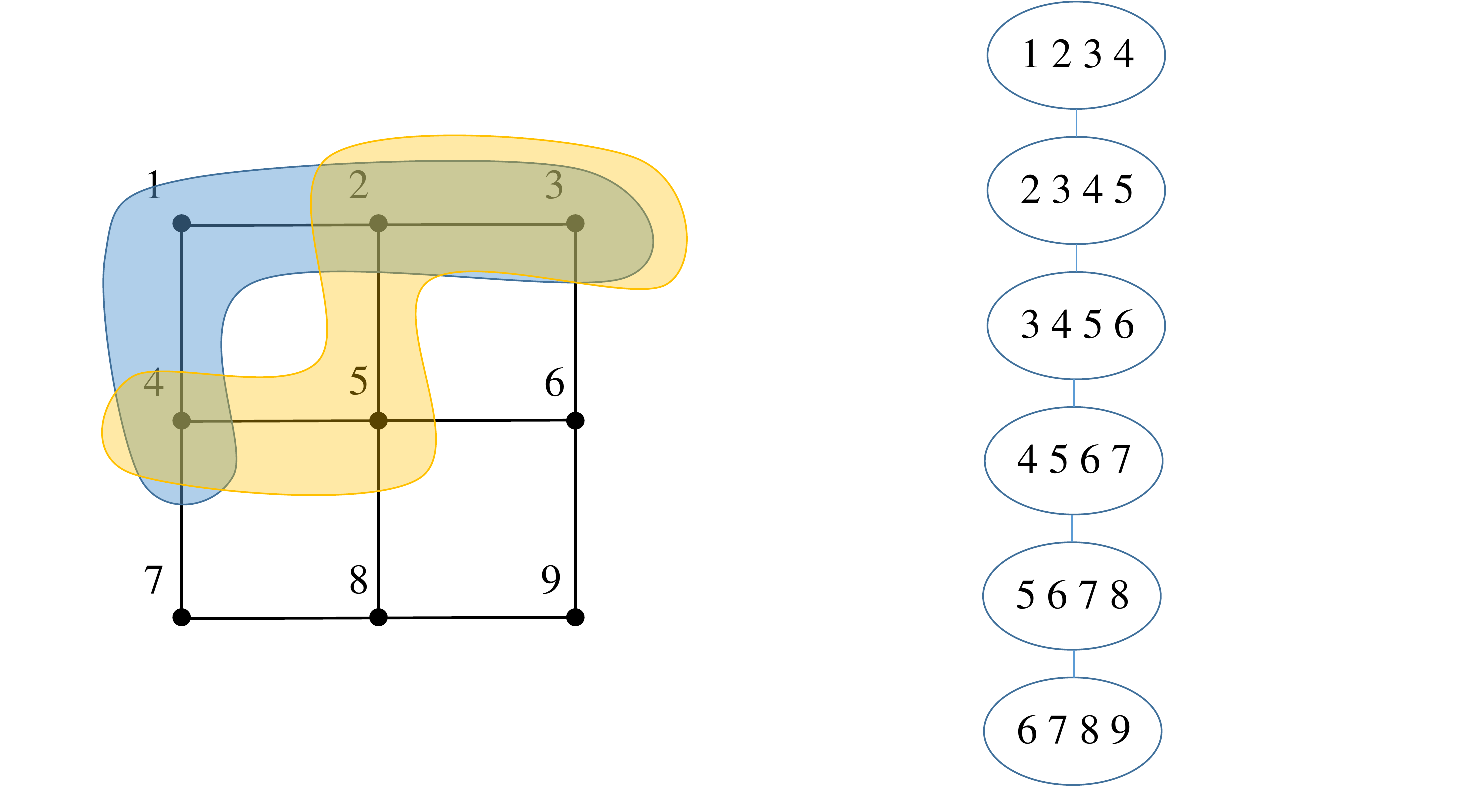}
\caption{$3 \times 3$ grid graph and its tree decomposition using band structure (left). The first bag (blue) has element $\{1,2,3,4\}$ and the second bag (yellow) has element $\{2,3,4,5\}$. The width of this tree decomposition is 3 (right), which is exactly equal to the treewidth.}
\label{fig:grid-tree}
\end{figure}

\section{Other initial configuration for boson sampling using local Haar-random beam splitters}
When we want to conduct an experiment or simulate for boson sampling that is believed to be hard \cite{aaronson2011computational}, we need to ensure collision-free cases, i.e., $M=\omega(N^2)$ or $\gamma>2$ in the main text.
We first emphasize that there is a possibility that the sampling for $M=O(N^2)$ is not as difficult as for $M=\omega(N^2)$ \cite{clifford2020faster,bulmer2021boundary}.
Let us focus on achieving such conditions with two representatives examples of different initial configurations in 2D systems, as shown in Fig.~\ref{fig:collision}. Here, we again assume local Haar-random circuits, which is motivated from the recent experimental setup \cite{zhong2020quantum, zhong2021phase}.
Also we assume that $\gamma$ is slightly larger than 2, which is the experimentally favorable and minimal condition to achieve collision-free cases.
We have shown that when $D=\Theta(N^{\gamma-\epsilon})$, any modes can interact with all modes regardless of an initial configuration.
Especially when the initial sources are equally distributed as we analyzed, the photons can be detected on all $M$ modes when $D=\Theta(D_\text{easy})=\Theta(N^{\gamma-1-\epsilon})$, i.e., all $M=\Theta(N^\gamma)$ modes are effectively involved.
Therefore, for example, when $\gamma$ is slightly greater than 2, we only need an almost linear depth of $N$ to achieve collision-free cases.
On the other hand, when all input modes are concentrated at the center or the corner of the $\sqrt{M}\times \sqrt{M}$ lattice as $\sqrt{N}\times \sqrt{N}$ sublattice as shown in Fig.~\ref{fig:collision} (a),
at $D=\Theta(D_\text{easy})$, photons can propagate up to $\Theta(N^{(\gamma-1)/2})$ and the number of effective modes is at most $\Theta(N^{\gamma-1})$.
Thus, a linear depth is not enough to achieve collision-free cases and we need a larger depth when $\gamma$ is slightly larger than 2.
To make the effective number of modes to be $M=\omega(N^2)$, the required depth is $D=\omega(N^{2-\epsilon})=\omega(N^{3-\gamma}D_\text{easy})$; the modes around edges do not effectively contribute before the depth.
A difference from the equally distributed case is that even though initial sources can interact each other at a smaller depth, since the initial sources are dense, it does not guarantee collision-free case when $D=O(N^{2-\epsilon})=O(N^{3-\gamma}D_\text{easy})$.
One can also show a similar behavior of another example shown in Fig.~\ref{fig:collision} (b). 
Finally, it is worth noting that if we do not assume local Haar-random circuits, there is an ensemble that reaches a global Haar-random and collision-free outcomes with a lower depth \cite{russell2017direct}.
Nevertheless, we emphasize again that a local Haar-random circuit is the experimental setup currently used for quantum supremacy demonstration, and that under this assumption, it is advantageous to choose the equally spacing initial state to minimize the depth for collision-free.


\begin{figure}[t]
\includegraphics[width=300px]{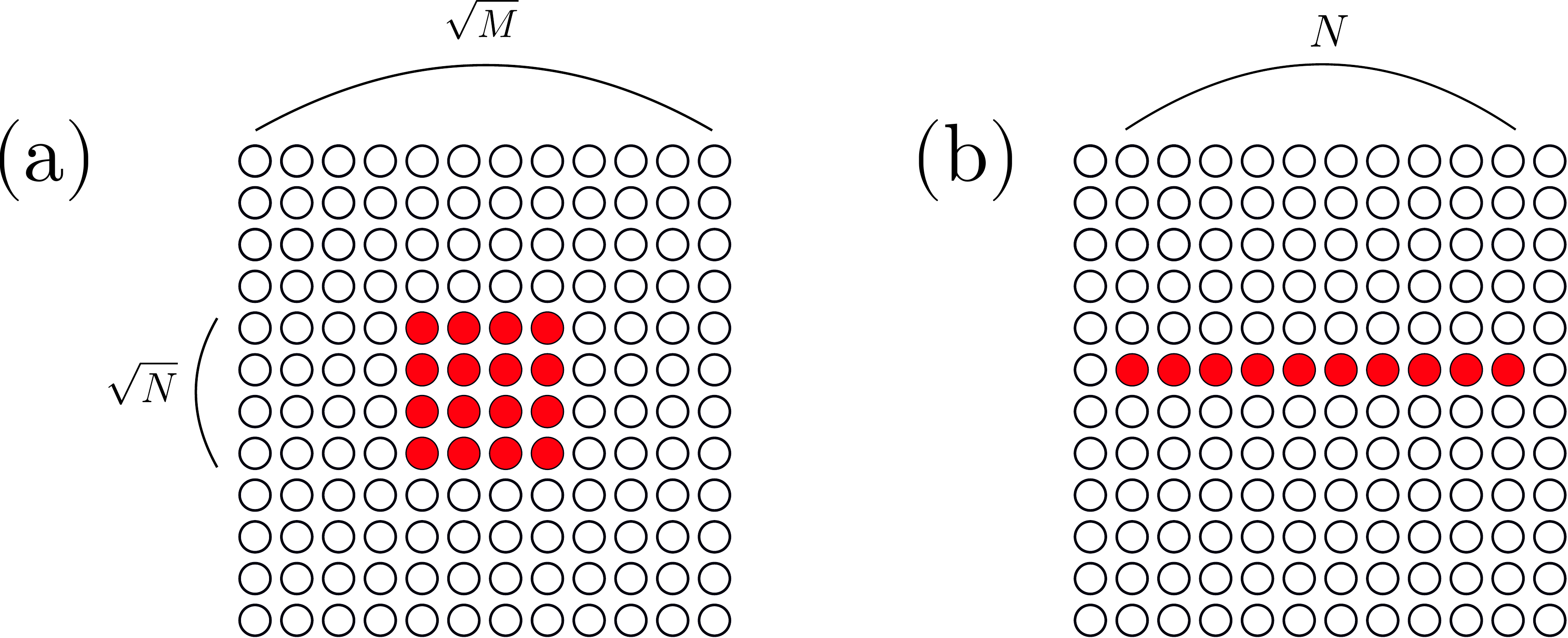}
\caption{Examples of different initial configurations. We assume $M=\Theta(N^2)$. (a) $\sqrt{N}\times\sqrt{N}$ grid shape of initial sources, (b) $1\times N$ strip shape of sources.}
\label{fig:collision}
\end{figure}

\section{Numerical method for Likelihood test}
In this section, we elaborate the numerical method that we have conducted to show that our treewidth-based algorithm gives a larger likelihood than GBS experiments.
The recent GBS experiment has used two layers of 2-dimensional beam-splitter arrays for 144 modes.
To reveal the locality of the given circuit, we first rearranged the 144 modes one-dimensionally as shown in Fig. 5(a).
In this procedure, we have randomly shuffled the input and output modes to find the best configuration for local approximation.
For each case, we approximate the unitary matrix of the circuit by imposing a spatial locality, i.e., photons from a source can propagate only for $K$ steps.
It means that we discard some elements of the matrix that correspond to the coupling between modes over than $K$ steps from each source.
We then choose a configuration that gives the best fidelity between an approximated Gaussian state with the propagation length $K$ and the ideal Gaussian state.
For this procedure, because the approximation effectively decreases the number of photons, we have increased the squeezing parameters and also thermal photons of the input state to have the same mean photon number as the experiment.
We note that the likelihood test is equivalent to the test implemented in Refs.~\cite{zhong2020quantum, zhong2021phase}.

More specifically, we compute the log-likelihood ratio:
\begin{align}
    \text{ratio}
    \equiv\log\frac{\text{Pr}_\text{ideal}(\text{Samples from experiment})}{\text{Pr}_\text{ideal}(\text{Samples from treewidth algorithm})}
    =\sum_{i=1}^{N_\text{samp}}\log\frac{\text{Pr}_\text{ideal}(\bm{m}^{(i)}_\text{exp})}{\text{Pr}_\text{ideal}(\bm{m}^{(i)}_\text{tree})},
\end{align}
where samples from experiments are written as $\{\bm{m}_\text{exp}^{(i)}\}_{i=1}^{N_\text{samp}}$, and samples from treewidth algorithm are written as $\{\bm{m}_\text{tree}^{(i)}\}_{i=1}^{N_\text{samp}}$, and $N_\text{samp}$ is the number of samples.

For the numerical demonstration, we have chosen the experimental data of focal waist $65\mu m$ and power $P=0.15,0.3,0.6,1.0,1.65~W$.
In the main text, we have presented the first two cases, which are classically verifiable since the classical simulation is possible.
For the larger powers, we have tested for the many different randomly chosen marginal modes and averaged the score, and the results are presented in Fig.~\ref{fig:HOG_SM}.
First, we observe how the score changes for the classically verifiable cases in Fig.~\ref{fig:HOG_SM}(a)(b).
While the overall behavior is not monotonic, the largest likelihood ratio is attained for the smallest number of marginal modes.
Now for the quantum supremacy regime, the behavior is almost consistent with the classically verifiable cases in the sense that the score decreases up to 60 marginal modes.
Due to the computational cost, we cannot test for larger number of marginals and ultimately for full distribution, but based on the classically verifiable cases, we expect that our classical algorithm may have a better likelihood than the experiment.
It is worth emphasizing that we had to choose a larger $K$ to attain a larger likelihood from our treewidth-based algorithm as the power grows, which is due to the fact that for a large number of photons, more photons can propagate further even under the same unitary matrix, so that the approximation error inevitably increases (e.g. Eq.~\eqref{eq:fid_bound}).

\begin{figure}[t]
\includegraphics[width=480px]{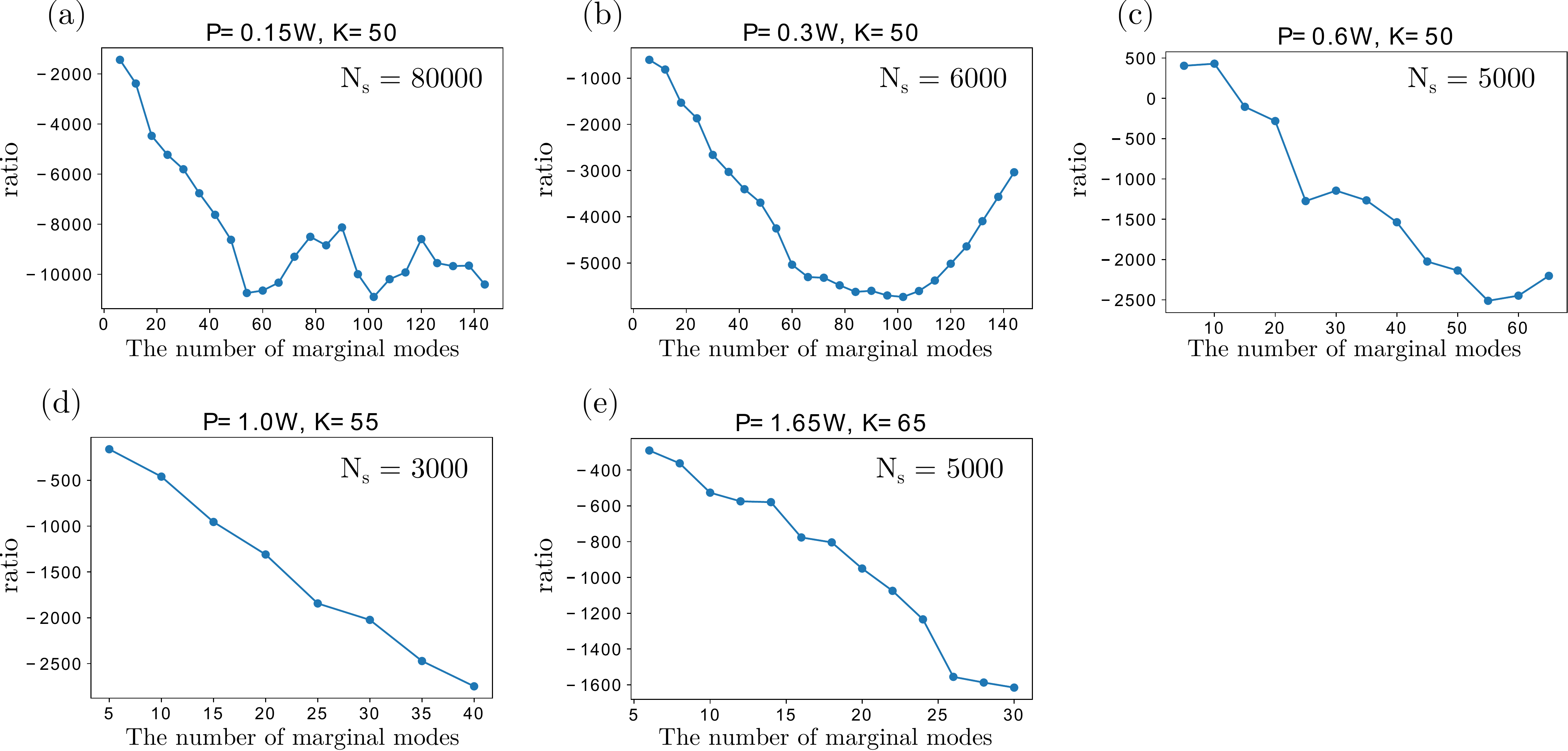}
\caption{Log-likelihood ratio for different powers and different number of marginal modes. We have averaged over 60, 20, 10, 30, 20 different choices of marginal modes for each point for (a)-(e), respectively.}
\label{fig:HOG_SM}
\end{figure}

\bibliography{reference}